\RequirePackage{amsmath}

\RequirePackage{fix-cm}
\documentclass[10pt]{llncs}
\usepackage{amsmath}
\usepackage{amssymb}
\usepackage{amsfonts}
\usepackage{url}
\usepackage[bookmarks,bookmarksdepth=2,pdfusetitle]{hyperref}
\usepackage{graphicx,wasysym}
\usepackage{lmodern}
\usepackage{paralist}
\usepackage{enumerate}
\usepackage[shortlabels]{enumitem}
\usepackage{bbm}
\usepackage{tikz}
\usepackage{verbatim}

\usepackage{comment}
\usepackage{booktabs}
\usepackage{lipsum}
\usepackage{tabularx}
\usepackage{anyfontsize}
\usepackage{fancyhdr}

\makeatletter
\newcommand*{\rom}[1]{\expandafter\@slowromancap\romannumeral #1@}
\usetikzlibrary{arrows}
\tikzset{
  treenode/.style = {align=center, inner sep=0pt, text centered,
    font=\sffamily},
    arn_n/.style = {treenode, circle,  draw=black, text width=1.5em,very thick},
  arn_r/.style = {treenode, circle,draw=black, text width=1.5em},
  arn_y/.style = {treenode, circle, draw=black, text width=1.8em}
}

\urldef{\mailsa}\path|{lingsan,hxwang,khoantt,zh0078ng}@ntu.edu.sg|
\AtBeginDocument{%
  \paperwidth=\dimexpr
    1in + \oddsidemargin
    + \textwidth
    + 1in + \oddsidemargin
  \relax
  \paperheight=\dimexpr
    1in + \topmargin
    + \headheight + \headsep
    + \textheight
    + 1in + \topmargin
  \relax
  \usepackage[pass]{geometry}\relax
}

\begin{document}

\newcommand{\0}{{\mathbf{0}}}
\newcommand{\A}{{\mathbf{A}}}
\newcommand{\U}{{\mathbf{U}}}
\renewcommand{\u}{{\mathbf{u}}}
\newcommand{\E}{{\mathbf{E}}}
\newcommand{\B}{{\mathbf{B}}}
\newcommand{\C}{{\mathbf{C}}}
\newcommand{\D}{{\mathbf{D}}}
\newcommand{\cD}{{\mathcal{D}}}
\newcommand{\CD}{{\mathcal{D}}}
\newcommand{\e}{{\mathbf{e}}}
\renewcommand{\c}{{\mathbf{c}}}
\newcommand{\f}{{\mathbf{f}}}
\newcommand{\F}{{\mathbf{F}}}
\newcommand{\bbe}{{\mathbbm{e}}}
\newcommand{\Ad}{{\mathcal{A}}}
\newcommand{\Fd}{\mathcal{F}}
\newcommand{\FF}{\mathbf{F}}
\newcommand{\GG}{\mathbf{G}}
\newcommand{\I}{{\mathcal{I}}}
\newcommand{\N}{\mathbb{N}}
\newcommand{\Bd}{{\mathcal{B}}}
\newcommand{\Sd}{{\mathcal{S}}}
\newcommand{\R}{{\mathbf{R}}}
\newcommand{\s}{{\mathbf{s}}}
\renewcommand{\d}{{\mathbf{d}}}
\newcommand{\T}{{\mathbf{T}}}
\newcommand{\Ts}{{\mathcal{T}}}
\newcommand{\x}{{\overrightarrow{x}}}
\newcommand{\y}{{\overrightarrow{y}}}
\newcommand{\Z}{{\mathbb{Z}}}
\newcommand{\z}{{\mathbf{z}}}
\newcommand{\w}{{\mathbf{w}}}
\newcommand{\W}{{\mathbf{W}}}
\newcommand{\rr}{{\mathbf{r}}}
\newcommand{\uu}{{\mathbf{u}}}
\newcommand{\V}{{\mathbf{V}}}
\renewcommand{\v}{{\mathbf{v}}}
\newcommand{\M}{{\mathbf{M}}}
\renewcommand{\S}{{\mathbf{S}}}
\newcommand{\ZZ}{{\mathbf{Z}}}
\newcommand{\HH}{{\mathbf{H}}}
\newcommand{\X}{{\mathbf{X}}}

\newcommand{\TrapGen}{\textsf{TrapGen}}
\newcommand{\KeyGen}{\mathsf{KeyGen}}
\newcommand{\Setup}{\mathsf{Setup}}
\newcommand{\Enc}{\mathsf{Enc}}
\newcommand{\Dec}{\mathsf{Dec}}
\newcommand{\UserKG}{\mathsf{UserKG}}
\newcommand{\Token}{\mathsf{Token}}
\newcommand{\UpdKG}{\mathsf{UpdKG}}
\newcommand{\TranKG}{\mathsf{TranKG}}
\newcommand{\Transform}{\mathsf{Transform}}
\newcommand{\Revoke}{\mathsf{Revoke}}
\newcommand{\Sys}{\mathsf{Sys}}
\newcommand{\Path}{\textsf{Path}}
\newcommand{\KUNodes}{\textsf{KUNodes}}
\newcommand{\sk}{\textsf{sk}}
\newcommand{\pk}{\textsf{pk}}
\newcommand{\RL}{\textsf{RL}}
\newcommand{\st}{\textsf{st}}
\newcommand{\params}{\textsf{params}}
\newcommand{\pp}{\textsf{pp}}
\newcommand{\msk}{\textsf{msk}}
\newcommand{\ct}{\textsf{ct}}
\newcommand{\BT}{\mathsf{BT}}
\newcommand{\bin}{\mathsf{bin}}
\newcommand{\id}{\textsf{id}}
\newcommand{\uk}{\textsf{uk}}
\renewcommand{\t}{\textsf{t}}
\renewcommand{\H}{\textsf{H}}
\newcommand{\tk}{\textsf{tk}}
\newcommand{\dk}{\textsf{dk}}

\newcommand{\SampleRwithBasis}{\textsf{SampleRwithBasis}}
\newcommand{\SamplePre}{\textsf{SamplePre}}
\newcommand{\HIBE}{\textsf{HIBE}}
\newcommand{\RIBE}{\textsf{RIBE}}
\newcommand{\PKE}{\textsf{PKE}}
\newcommand{\PE}{\textsf{PE}}
\newcommand{\SampleLeft}{\textsf{SampleLeft}}
\newcommand{\SampleBLeft}{\textsf{SampleBasisLeft}}
\newcommand{\SampleRight}{\textsf{SampleRight}}
\newcommand{\SampleBRight}

\title{Server-Aided Revocable Predicate Encryption:\\Formalization and Lattice-Based Instantiation}
\author{San Ling,
 Khoa Nguyen,
 Huaxiong Wang,
 Juanyang Zhang
}
\institute{Division of Mathematical Sciences, School of Physical and Mathematical Sciences,\\
Nanyang Technological University, Singapore\\
\mailsa}
\maketitle

\begin{abstract}
Efficient user revocation is a necessary but challenging problem in many multi-user cryptosystems. Among known approaches, server-aided revocation yields a promising solution, because it allows to outsource the major workloads of system users to a computationally powerful third party, called the server, whose only requirement is to carry out the computations correctly. Such a revocation mechanism was considered in the settings of identity-based encryption and attribute-based encryption by Qin et al. (ESORICS 2015) and Cui et al. (ESORICS 2016), respectively.

In this work, we consider the server-aided revocation mechanism in the more elaborate setting of predicate encryption (PE). The latter, introduced by  Katz, Sahai, and Waters (EUROCRYPT 2008), provides fine-grained and role-based access to encrypted data and can be viewed as a generalization of identity-based and attribute-based encryption.
Our contribution is two-fold. First, we formalize the model of server-aided revocable predicate encryption (SR-PE), with  rigorous definitions and security notions. Our model can be seen as a non-trivial adaptation of Cui et al.'s work into the PE context. Second,
% which inherits the efficiency advantage of the server-aided method and the attribute-hiding property of predicate encryption.
we put forward a lattice-based instantiation of SR-PE. The scheme employs the PE scheme of Agrawal, Freeman and Vaikuntanathan (ASIACRYPT 2011) and the complete subtree method of Naor, Naor, and Lotspiech (CRYPTO 2001) as the two main ingredients, which work smoothly together thanks to a few additional techniques. Our scheme is proven secure in the standard model (in a selective manner), based on the hardness of the Learning With Errors $(\mathsf{LWE})$ problem.
%\subclass{94A60}
%\begin{keywords}
%Server-aided revocation, predicate encryption, lattices, LWE
%\end{keywords}
\end{abstract}

\section{Introduction}\label{section: Introduction}
The notion of predicate encryption (PE), formalized by Katz, Sahai, and Waters~\cite{KatzSW08}, is an emerging paradigm of public-key encryption, which provides fine-grained and role-based access to encrypted data. In a PE scheme, the user's private key, issued by the key generation center (KGC), is associated with a predicate~$f$, while a ciphertext is bound to an attribute~$I$.
Then the system ensures that the user can decrypt the ciphertext if and only if $f(I)=1$.
PE can be viewed as a generalization of attribute-based encryption (ABE)~\cite{SahaiW05,GoyalPSW06}.
Whereas the latter reveals the attribute bound to each ciphertext, the former preserves the privacy of not only the encrypted data but also the attribute. These powerful properties of PE yield numerous potential applications (see, e.g., \cite{BonehW07,ShiBCSP07,KatzSW08}).

As for many other multi-user cryptosystems, an efficient revocation mechanism is necessary and imperative in the PE setting. When some users misbehave or when their private keys are compromised, the users should be revoked from the system and we would need a non-trivial mechanism ensuring that: (i) revoked users can no longer decrypt ciphertexts; (ii) the workloads of the KGC and the non-revoked users in updating the system are not too taxing.
In the ABE setting, Boldyreval et al.~\cite{BoldyrevaGK08} put forward a revocation mechanism based on a time-based key update procedure.  In their approach, a ciphertext is not only bound to an attribute but also to a time period. The KGC, who possesses the up-to-date list of revoked users, has to publish an update key at each time period so that only non-revoked users can update their private keys to decrypt ciphertexts bound to the same time slot.  To make the KGC's workload scalable (i.e., being logarithmic in the maximum number of users~$N$), Boldyreval et al. suggested to employ the subset-cover framework due to Naor et al.~\cite{NaorNL01} to handle key updating. Concrete pairing-based instantiations of revocable ABE following this approach were proposed in~\cite{Attrapadung2009,SahaiSW12}.

In Boldyreval et al.'s model, however, the non-revoked users have to communicate with the KGC regularly to receive the update keys. Although such key updating process can be done through a public channel, it is somewhat inconvenient and bandwidth-consuming. To reduce the users's computational burden, Qin et al.~\cite{SRIBE} proposed an interesting solution in the context of identity-based encryption (IBE), called server-aided revocable identity-based encryption (SR-IBE). Qin et al.'s model takes advantage of a publicly accessible server with powerful computational capabilities, to which one can outsource  most of users' workloads.  Moreover, the server can be untrusted in the sense that it does not possess any secret information.

Cui et al.~\cite{CuiDLQ16} subsequently adapted the server-aided revocation mechanism into the ABE setting and introduced server-aided revocable attribute-based encryption (SR-ABE). Briefly speaking, an SR-ABE scheme works as follows.
When a new user joins the system, he generates a public-secret key-pair, and sends the public key to the KGC~\footnote{Alternatively, as pointed out by Cui et al.~\cite{CuiDLQ16}, the key-pair can be generated by the KGC and then sent to the user. This requires a secure channel - which is typically assumed to be available in the setting of centralized cryptosystems.}. The latter then generates
a user-specific token that is forwarded to the untrusted server through a public channel.
%When a user joins the system, the KGC generates a pair of private and public key and sends the key pair to him through a secure channel. \footnote{In~\cite{CuiDLQ16}, the private and public key pair can also be generated by the user himself and then the public key is sent to the KGC via a public channel.}
%Meanwhile, the KGC generates a token associated with his public key and predicate, and sends this token to the untrusted server through a public channel.
At each time period, the update key is sent only to the server rather than to all users. To perform decryption for a specific user, the server first transforms the ciphertext into a ``partially decrypted ciphertext''.  The latter is bound to the user's public key, so that only the intended user can recover the plaintext using his private key. In~\cite{CuiDLQ16}, apart from introducing this new model, Cui et al. also described a pairing-based instantiation of SR-ABE.

In this work, inspired by the potentials of PE and the advantages of the server-aided revocation mechanism, we consider the notion of sever-aided revocable predicate encryption, and aim to design the first such scheme from lattice assumptions.

\medskip

\noindent
{\sc Other related works. }
 The subset-cover framework, proposed by Naor et al.~\cite{NaorNL01},
 is arguably the most well-known revocation technique for multi-user
systems. It uses a binary tree, each leaf of which is designated to each user. Non-revoked users are partitioned
into disjoint subsets, and are assigned keys according to the complete subtree (CS) method or the subset
difference (SD) method. This framework was considered to address user revocation for identity-based encryption (IBE) schemes, with constructions from pairings~\cite{BoldyrevaGK08,LibertV09,SeoE13,LeeLP14} and from lattices~\cite{ChenLLWN12,ChengZhang2015,Takayasu017}. It also found applications in the context of revocable group signatures~\cite{LibertPY12a,LibertPY12b} and revocable ABE schemes~\cite{BoldyrevaGK08,Attrapadung2009,SahaiSW12}.

Lattice-based cryptography, pioneered by Ajtai~\cite{Ajtai96}, Regev~\cite{Regev05} and
Gentry et al.~\cite{GPV08}, has been an exciting research area in the last decade, providing several advantages over conventional number-theoretic cryptography, such as faster arithmetic operations and conjectured resistance against quantum computers. Among other primitives, lattice-based revocable cryptosystems have been receiving considerable attention.

Chen et al.~\cite{ChenLLWN12} initiated the study of lattice-based revocable IBE, equipping the scheme by Agrawal, Boneh and Boyen~\cite{ABB10} with a revocation method following Boldyreval et al.'s blueprint~\cite{BoldyrevaGK08}. Chen et al.'s construction has been improved in two directions. Nguyen et al.~\cite{NguyenWZ16} extended it into an SR-IBE scheme, using a hierarchical IBE~\cite{ABB10} and a double encryption technique~\cite{DoubleEnc} in the process.
Takayasu and Watanabe~\cite{Takayasu017} developed a scheme with enhanced security, which is, to some extent, resistant against decryption key exposure attacks~\cite{SeoE13}.
Very recently, a major related result was obtained by
Agrawal et al.~\cite{AgrawalBPS017}, who established a lattice-based identity-based trace-and-revoke system, via
an elegant generic construction from functional encryption for inner products.

Beyond the IBE setting, Ling et at.~\cite{LingNWZ17} provided a revocation method for the lattice-based PE scheme by Agrawal, Freeman and Vaikuntanathan~\cite{AgrawalFV11}. To this end, Ling et al. employs the direct revocation approach~\cite{NietoMS12},
where the revocation information is directly embedded into each ciphertext. This approach eliminates the necessity of the key-update phase, but it produces ciphertexts of relatively large size, depending on
the number of all users $N$ and/or the number of revoked users $r$. For the time being, the problem of constructing lattice-based revocable PE schemes featuring constant-size ciphertexts remains open.

\medskip

\noindent
{\sc Our results and techniques. } The contribution of this work is two-fold: We first formalize the concept of server-aided predicate encryption (SR-PE), and then put forward an instantiation of SR-PE from lattices. An overview of these two results is given below.

Our model of SR-PE inherits the main advantage of the server-aided revocation mechanism~\cite{SRIBE}: most of the users' workloads are delegated to an untrusted server. The model can be seen as a non-trivial adaptation of Cui et al.'s model of SR-ABE~\cite{CuiDLQ16} into the PE setting, with two notable distinctions. First, while Cui et al. assume a public-secret key-pair for each user, we do not require users to maintain their own public keys.
%- a setting that, in our opinion, is more compatible with the spirit of identity-based/attribute-based/predicate cryptosystems.
Recall that Shamir's~\cite{Shamir84} motivation to initiate the study of IBE is to eliminate the burden of managing public keys. The more general notions of ABE and PE later inherit this advantage over traditional public key encryption. From this point of view, the re-introduction of users' public keys seems contradict to the spirit of identity-based/attribute-based/predicate cryptosystems. Thus, by not demanding the existence of users' public keys, we make our model consistent with ordinary (i.e., non-revocable) predicate encryption.
Second, our security definition reflects the attribute-hiding property of PE systems, which guarantees that attributes bound to the ciphertexts are not revealed during decryptions, and which is not considered in the context of ABE.

%Formally, we define the security notion for SR-PE schemes: the attribute hiding security against chosen plaintext attacks.

As an effort to instantiate a scheme satisfying our model under post-quantum assumptions, we design a lattice-based construction that is proven secure (in a selective manner) in the standard model, assuming the hardness of the Learning With Errors ($\mathsf{LWE}$) problem~\cite{Regev05}.
The efficiency of our scheme is comparable to that of the constructions under the server-aided revocation approach~\cite{SRIBE,CuiDLQ16,NguyenWZ16}, in the following sense. The sizes of private keys and ciphertexts, as well as the complexity of decryption on the user side are all independent of the number of users $N$ and the number of revoked users $r$. In particular, the ciphertext size in our scheme compares favourably to that of Ling et al.'s revocable PE scheme~\cite{LingNWZ17}, which follows the direct revocation approach. It is also worth mentioning that, if we do not assume the availability of the server (which does not affect security because the server does not possess any secret key) and let the users perform the server's work themselves, then our scheme would yield the first (non-server-aided) lattice-based PE with constant-size ciphertexts.
%Additionally, the use of the server does not sacrifice the security, which means that the scheme is still secure if users themselves perform server's computations.
At a high level, our scheme employs two main building blocks:
the ordinary PE scheme by Agrawal et al.~\cite{AgrawalFV11} and the CS method due to Naor et al.~\cite{NaorNL01}.
We observe that the same two ingredients were adopted in Ling et at.'s scheme~\cite{LingNWZ17}, but their direct revocation approach is fundamentally different from ours, and thus, we have to find a new way to make these ingredients work smoothly together.

Our first challenge is to enable a relatively sophisticated mechanism, in which an original PE ciphertext bound to an attribute and a time period (but not bound to any user's identifying information), after being transformed by the server, would become a partially decrypted ciphertext bound to the identifying information of the non-revoked recipient. We note that, in the setting of lattice-based SR-IBE, Nguyen et al.~\cite{NguyenWZ16} addressed a somewhat related problem using a double encryption technique, where the original and the partially decrypted ciphertexts are both bound to the recipient's identity and time period. However such technique requires the sender to know the recipient's identity when generating the ciphertext, and hence, it is not applicable to the PE setting. We further note that, Cui et al.~\cite{CuiDLQ16} solved a more closely related problem, in which the partially decrypted ciphertext is constrained to bind to the recipient's public key - with respect to some public-key encryption (PKE) system. We observe that it is possible to adapt the technique from~\cite{CuiDLQ16}, but as our SR-PE model does not work with users' public keys, we will instead make use of an IBE instance. Namely, we additionally employ the IBE from~\cite{ABB10} and assign each user an identity $\id$. The challenge now is how to embed $\id$ into the user-specific token in a way such that the partially decrypted ciphertext will be bound to $\id$.

To address the above problem, we exploit a special property of some \textsf{LWE}-based encryption systems, observed by Boneh et al.~\cite{BonehGGHNSVV14}, which allows to transform an encryption of a message under one key into an encryption of the same message under another key.
%This enables us to evaluate the ciphertetxt to be bound to any matrix if the key is specified by this matrix as well.
Then, our scheme works roughly as follows.
%To encrypt $M$, the sender generates an PE ciphertext~\cite{AgrawalFV11} of $M$, that contains an additional component bound with a time period. Meanwhile, the token for recipient with identity $\id$ is issued using a PE instance and specified by an $id$-determined matrix $\D_\id$.
%Our lattice-based SR-PE scheme can be seen as a combination of two PE systems~\cite{AgrawalFV11} and the CS method~\cite{NaorNL01}.
Each user with identity $\id$ is issued a private key for a two-level hierarchical system consisting of one instance of the PE system from~\cite{AgrawalFV11} as well as an additional IBE level for~$\id$, associated with a matrix $\D_\id$.
Meanwhile, the token for $\id$ is generated by embedding $\D_\id$  into another instance of the same PE system~\cite{AgrawalFV11}.
%Meanwhile, the KGC assigns $\id$ to a leaf node of a binary tree, records the nodes in the tree path corresponding to $\id$, and generates the token for $\id$ by embedding an IBE component into another PE instance.
At each time period $\t$, the KGC employs the CS method to compute an update key $\uk_\t$ and sends it to the server.
A ciphertext in our scheme is a combination of two PE ciphertexts and an extra component bound to $\t$.
If recipient $\id$ is not revoked at time period $\t$, the server can use the token for $\id$ and $\uk_\t$ to transform the second PE ciphertext into an IBE ciphertext associated with $\D_\id$, thanks to the special property mentioned above.
Finally, the partially decrypted ciphertext, consisting of the first PE ciphertext and the IBE ciphertext, can be fully decrypted using the private key of $\id$.

The security of our proposed SR-PE scheme relies on that of the two lattice-based components from~\cite{AgrawalFV11} and~\cite{ABB10}. Both of them are selectively secure in the standard model, assuming
the hardness of the \textsf{LWE} problem - so is our scheme.

\medskip

\noindent
{\sc Organization. } The rest of this paper is organized as follows. In Section \ref{section: background}, we briefly recall some background about lattices and the CS method. We give the rigorous definitions and security model of SR-PE in Section~\ref{section: model}. Our lattice-based instantiation of SR-PE is described in Section~\ref{section: main scheme} and analyzed in Section~\ref{section: analysis}. Finally, Section~\ref{section: collusion} concludes the paper.

\section{Preliminaries} \label{section: background}
{\sc Notations. }The acronym PPT stands for ``probabilistic polynomial-time''. %We say that a function~$d:{\mathbb N}\rightarrow{\mathbb R}$ is negligible, if it is $O\left(n^{-c}\right)$ for all $c>0$.
%The statistical distance of two random variables $X$ and $Y$ over a discrete domain $\Omega$ is defined as~$\Delta(X;Y):= \frac{1}{2}\sum_{s\in\Omega}\left|\Pr[X=s]-\Pr[Y=s]\right|$. If~$X(\lambda)$ and~$Y(\lambda)$ are ensembles of random variables, we say that~$X$ and~$Y$ are statistically close if~$\Delta(X(\lambda);Y(\lambda))$ is a negligible function of~$\lambda$.
We often write~$x\hookleftarrow\chi$ to indicate that we sample~$x$ from probability distribution $\chi$. If~$\Omega$ is a finite set, the notation~$x\stackrel{\$}{\leftarrow}\Omega$ means that~$x$ is chosen uniformly at random from~$\Omega$. Meanwhile, if $x$ is an output of PPT algorithm $\mathcal{A}$, then we write $x \leftarrow \mathcal{A}$.

We use bold upper-case letters (e.g., $\A,\B$) to denote matrices and use bold lower-case letters (e.g., $\mathbf{x},\mathbf{y}$) to denote column vectors. In addition, we user over-arrows to denote predicate and attribute vectors as $\x,\y$.  For two matrices $\A\in\mathbb{R}^{n\times m}$ and $\B\in\mathbb{R}^{n\times k}$, we denote by ~$\left[\A \mid \B\right]\in\mathbb{R}^{n\times (m+k)}$ the column-concatenation of $\A$ and $\B$. For a vector $\mathbf{x}\in\Z^{n}$, $||\mathbf{x}||$ denotes the Euclidean norm of $\mathbf{x}$. We use~$\widetilde{\A}$ to denote the Gram-Schmidt orthogonalization of matrix $\A$, and $||{\A}||$ to denote the Euclidean norm of the longest column in $\A$.  If $n$ is a positive integer,~$[n]$ denotes the set~$\{1,..,n\}$.
For~$c\in \mathbb{R}$, let $\lfloor c \rceil =\lceil c-1/2 \rceil$ denote the integer closest to~$c$.

\subsection{Background on Lattices}\label{subsection:lattice-background}
%In this section, we describe the concepts that we will need for our construction and the security proof.
\noindent
{\bf Integer lattices. }
An  {\em $m$-dimensional  lattice} $\Lambda$ is a  discrete subgroup of $\mathbb{R}^m$. A full-rank matrix $\B\in\mathbb{R}^{m\times m}$ is a {\em  basis} of $\Lambda$ if $\Lambda=\{ \mathbf{y}\in\mathbb{R}^m: \exists \hspace*{1pt}\s\in\Z^m,\mathbf{y}=\B \cdot \s  \}.$
 We consider integer lattices, i.e., when $\Lambda\subseteq\Z^m$. For any integer $q\geq 2$ and any  $\mathbf{A} \in \mathbb{Z}_q^{n \times m}$, define the $q$-ary lattice:
$$\Lambda^\bot_q(\mathbf{A})= \big\{\mathbf{r} \in \mathbb{Z}^m: \hspace{7pt} \mathbf{A}\cdot \mathbf{r}= \mathbf{0} \bmod q\big\} \subseteq \mathbb{Z}^m.$$
For any $\mathbf{u}$ in the image of $\mathbf{A}$, define $\Lambda_{q}^{\mathbf{u}}(\mathbf{A})= \big\{\mathbf{r} \in \mathbb{Z}^m:\hspace{3pt} \mathbf{A}\cdot \mathbf{r}= \mathbf{u} \bmod q\big\}$.

A fundamental tool in lattice-based cryptography is an algorithm that generates a matrix $\mathbf{A}$ close to uniform together with a short basis $\T_{\A}$ of $\Lambda^{\bot}_q(\A)$.
\begin{lemma}[{\cite{Ajtai99,AP09,MicciancioP12}}]\label{Lemma:TrapGen}
   Let $n\geq 1,q\geq 2$ and $m\geq 2n  \log q $ be integers. There exists a PPT algorithm $\mathsf{TrapGen}(n,q,m)$ that outputs a pair $(\A, \T_\A)$ such that $\A$ is statistically close to uniform over $\mathbb{Z}_q^{n\times m}$ and $\T_\A \in \mathbb{Z}^{m\times m}$ is a basis for $\Lambda^{\bot}_q(\A)$, satisfying $$\|\mathbf{\widetilde{ T_A}}\|\leq O(\sqrt{n \log q}) \text{ and } \| \T_\A \| \leq O(n \log q)$$
   with all but negligible probability in $n$.
 \end{lemma}

 Micciancio and Peikert~\cite{MicciancioP12} consider a structured matrix $\GG$, called the \emph{primitive matrix}, which admits a publicly known short basis $\T_{\GG}$ of $\Lambda^{\bot}_q(\GG)$.

 %To simplify the notation, we will always assume that $\GG$ has the same dimension as the matrix $\A$ output by the algorithm $\TrapGen$.
%We recall the property of the primitive matrix $\GG$ defined in~\cite{MicciancioP12}.
\begin{lemma}[\cite{MicciancioP12}]\label{lemma:primitive matrix}
 Let $n\geq 1,q\geq 2$ be integers and let $m \geq n \lceil \log q \rceil $. There exists a full-rank matrix  $\GG \in \Z_q^{n\times m}$ such that the lattice $\Lambda_q ^{\bot}(\GG)$ has a known basis $\T_{\GG}\in\Z^{m\times m}$ with $||\widetilde{\T_{\GG}}||\leq \sqrt{5}$.

Furthermore, there exists a deterministic polynomial-time algorithm $\GG^{-1}$ which takes the input $\U\in\Z_q^{n\times m}$ and outputs $\mathbf{X}=\GG^{-1}(\U)$ such that $\mathbf{X}\in\{0,1\}^{m\times m}$ and $\GG\mathbf{X}=\U$.
\end{lemma}

\smallskip

\noindent
{\bf Discrete Gaussians. }
Let $\Lambda$ be an integer lattice. For vector $\c\in\mathbb{R}^m$ and any parameter $s>0$, define
     $\rho_{s,\c}(\rr)=\exp(-\pi\dfrac{\|\rr-\c\|^2}{s^2})$ and $\rho_{s,\c}(\Lambda)=\sum_{\rr\in \Lambda}\rho_{s,\c}(\rr).$
The  discrete Gaussian distribution over $\Lambda$ with center $\c$ and parameter $s$ is $\forall \rr\in \Lambda,\CD_{\Lambda,s,\c}(\rr)=\dfrac{\rho_{s,\c}(\rr)}{\rho_{s,\c}(\Lambda)}.$
 If $\c=\0$, we simplify to use notations $\rho_{s}$ and $\CD_{\Lambda,s}$.

We also need the following lemma to prove the correctness of construction in Section~\ref{section: main scheme}.
 \begin{lemma}[{\cite{GPV08,MicciancioP12,GayMW15,LingNWZ17}}]\label{lemma:distribution}
 Let $n\geq 1,q\geq 2$, $m\geq 2n  \log q $ and $k \geq 1$ be integers. Let $\mathbf{F}$ be a full-rank matrix in $\Z_q^{n\times m}$ and $\T_{\mathbf{F}}$ be a basis of $\Lambda^\bot_q (\mathbf{F})$. Assume that $s \geq ||\widetilde{\T_{\mathbf{F}}}|| \cdot\omega(\sqrt{\log n})$.  Then, for $\ZZ\hookleftarrow\left(\CD_{\Z^m,s}\right)^k$, the distribution of $\mathbf{F}\ZZ \bmod q$ is statistically close to the uniform distribution over $\Z_q^{n\times k}$.

\end{lemma}

%Particularly, Lemma~\ref{lemma:distribution} holds when $\mathbf{F}$ is the primitive matrix $\GG\in\Z_q^{n\times m}$ specified in Lemma~\ref{lemma:primitive matrix}.

\smallskip

\noindent
{\bf Sampling algorithms. }
It was shown in~\cite{ABB10,MicciancioP12} how to efficiently sample short vectors from specific lattices. Algorithms $\SamplePre$, $\mathsf{SampleBasisLeft}$, $\SampleLeft$ and $\SampleRight$ from those works will be employed in our construction and the security proof.
\begin{description}
\item[$\mathsf{SamplePre}$]\hspace{-3.5pt}$(\A,\hspace{2pt}\mathbf{T_A},\hspace{2pt}\u,\hspace{2pt}s)$: On input a full-rank matrix $\A \in \mathbb{Z}^{n\times m}_q$, a trapdoor $\mathbf{T_A}$ of $\Lambda_q^\bot(\A)$, a vector $\u\in\mathbb{Z}^{n}_q$, {and} a Gaussian parameter $s\geq \|\widetilde{\mathbf{T_A}}\|\cdot \omega(\sqrt{\log m})$, it outputs a vector $\e\in\mathbb{Z}^{m}$ with a distribution statistically close to $\mathcal{D}_{\Lambda^{\u}_q(\A),s}$.\smallskip
\item[$\mathsf{SampleBasisLeft}$]\hspace{-5pt}$(\A,\hspace{2pt}\mathbf{M},\hspace{2pt}\mathbf{T_A},\hspace{2pt}s)$: On input a full-rank matrix $\A \in \mathbb{Z}^{n\times m}_q$, a matrix $\mathbf{M} \in \mathbb{Z}^{n\times m}_q$, a trapdoor $\mathbf{T_A}$ of $\Lambda_q^\bot(\A)$ {and} a Gaussian parameter $s\geq \|\widetilde{\mathbf{T_A}}\|\cdot \omega(\sqrt{\log 2m})$, it outputs a basis $\T_{\FF}$ of $\Lambda^{\bot}_q(\FF)$, where $\FF=[\A\hspace*{1.5pt}|\hspace*{1.5pt}\M] \in \mathbb{Z}_q^{n \times 2m}$ and $||\widetilde{\T_\FF}||=||\widetilde{\T_\A}||$.\smallskip
\item[$\mathsf{SampleLeft}$]\hspace{-5pt}$(\A,\hspace{2pt}\mathbf{M},\mathbf{T_A},\hspace{2pt}\u,\hspace{2pt}s)$: On input full-rank matrix $\A \in \mathbb{Z}^{n\times m}_q$, a matrix $\mathbf{M} \in \mathbb{Z}^{n\times m}_q$, a trapdoor $\mathbf{T_A}$ of $\Lambda_q^\bot(\A)$,
a vector $\u \in\mathbb{Z}^{n}_q$, {and} a Gaussian parameter $s\geq \|\widetilde{\mathbf{T_A}}\|\cdot \omega(\sqrt{\log 2m})$, it outputs a vector $\mathbf{z} \in\mathbb{Z}^{2m}$, which is sampled from a distribution statistically close to $\mathcal{D}_{\Lambda^{\mathbf{u}}_q(\FF),s}$. Here we define $\FF=\left[\A\hspace*{1.5pt}|\hspace*{1.5pt}\M\right] \in \mathbb{Z}_q^{n \times 2m}$.\smallskip
\item[$\mathsf{SampleRight}$]\hspace{-5pt}$(\A,\hspace{2pt}\R,\hspace{2pt}\mathbf{G},\hspace{2pt}\mathbf{T_G},\hspace{2pt}\u,\hspace{2pt}s)$: On input matrices $\A \in \mathbb{Z}^{n\times m}_q, \R\in\Z^{m\times m}$,  the primitive matrix $\GG \in \mathbb{Z}^{n\times m}_q$ together with trapdoor $\mathbf{T_G}$ of $\Lambda_q^\bot(\GG)$,
a vector $\u \in\mathbb{Z}^{n}_q$, {and} a Gaussian parameter $s\geq \|\widetilde{\mathbf{T_B}}\|\cdot ||\R||\cdot \omega(\sqrt{\log m})$, it outputs a vector $\mathbf{z} \in\mathbb{Z}^{2m}$, sampled from a distribution statistically close to $\mathcal{D}_{\Lambda^{\mathbf{u}}_q(\FF),s}$. Here we define $\FF=\left[\A\hspace*{1.5pt}|\hspace*{1.5pt}\A\R+\GG\right] \in \mathbb{Z}_q^{n \times 2 m}$.
\end{description}

The above algorithms can be easily extended to the case of taking a matrix $\mathbf{U} \in \mathbb{Z}_q^{n \times k}$, for some $k \geq 1$. Then, the output is a matrix $\mathbf{Z} \in \mathbb{Z}$ with $k$ columns.

%We will also need a variant of left over hash lemma from~\cite{ABB10}.
%\begin{lemma}\label{lemma:distribution of R}
%  Suppose that $m>(n+1)\log q+\omega(\log n)$ and $q>2$ is a prime. Choose $\A \xleftarrow{\$}\Z^{n\times m}_q$, $\B\xleftarrow{\$}\Z^{n\times \kappa}_q$ and $\R\xleftarrow{\$}\{-1,1\}^{m\times \kappa}$ where $\kappa=\kappa(n)$ is polynomial in $n$. Then for any vector $\mathbf{v}\in \Z_q^m$, the distribution of $(\A,\A\R,\R^\top\mathbf{v})$ is statistically close to the distribution of $(\A,\B,\R^\top\mathbf{v})$.
%\end{lemma}

\vspace{.4cm}

\noindent
{\bf Learning With Errors. } We now recall the Learning With Errors ($\mathsf{LWE}$) problem~\cite{Regev05}, as well as its hardness.

\begin{definition}[$\mathsf{LWE}$] \label{definition:LWE} Let $n, m \geq 1,q \geq 2$, and let $\chi$ be a probability distribution on $\Z$. For $\s \in\mathbb{Z}^{n}_q$, let $\A_{\s,\chi}$ be the distribution obtained by sampling $\mathbf{a}\stackrel{\$}{\leftarrow} \mathbb{Z}^{n}_q$ and $e \hookleftarrow \chi$, and outputting
the pair $\left(\mathbf{a},\mathbf{a}^\top\s + e\right) \in \Z_q^n\times\Z_q$. The $(n,q,\chi)$-$\mathsf{LWE}$ problem asks to distinguish $m$ samples chosen
according to $\A_{\s,\chi}$ (for $\s \stackrel{\$}{\leftarrow} \mathbb{Z}^{n}_q$) and $m$ samples chosen according to the uniform distribution
over $\Z_q^n\times\Z_q$.
\end{definition}

If $q$ is a prime power and $B \geq\sqrt{n}\cdot \omega\left(\log n\right)$, then there exists an efficient sampleable
$B$-bounded distribution $\chi$ (i.e., $\chi$ outputs samples with norm at most $B$ with overwhelming
probability) such that $(n,q,\chi)$-$\mathsf{LWE}$ is as least as hard as worst-case lattice problem \textnormal{SIVP} with approximate factor~$O\left( nq/B \right)$ (see, e.g.,~\cite{Regev05,Pei09,MicciancioP12}).

\subsection{The Agrawal-Freeman-Vaikuntanathan Predicate Encryption Scheme}\label{subsection:PE}
Next, we recall the LWE-based predicate encryption, proposed by Agrawal, Freeman and Vaikuntanathan (AFV)~\cite{AgrawalFV11} and improved by~Xagawa~\cite{Xagawa13}. The scheme is for inner-product predicates, where an attribute is expressed as a vector $\y\in \Z_q^{\ell}$ (for some integers $q$ and $\ell$) and a predicate $f_{\x}$  is associated with a vector $\x\in \Z_q^{\ell}$. We say that $f_{\x}(\y)=1$ if $\langle \x,\y \rangle=0$, and $f_{\x}(\y)=0$ otherwise.
%The set $\mathbb{A}=\Z_q^{\ell}$ is called the attribute space, while the set $\mathbb{P}=\{f_{\x}\hspace*{1.8pt}\big| \hspace*{1.8pt}\x\in\Z_q^{\ell}\}$ is called the predicate space. %If $N$ is the maximal number of keys in the system, then we call the set $\mathcal{I}=[N]$ the index space.

In the AFV scheme, the key authority possesses a short basis $\T_{\A}$ for a public lattice $\Lambda_q^{\bot}(\A)$, generated by  $\TrapGen$ algorithm. Each predicate vector $\x$ is associated with a super-lattice of $\Lambda_q^{\bot}(\A)$, a short vector of which can be efficiently computed using the trapdoor~$\T_{\A}$. Such a short vector allows  to decrypt a Dual-Regev ciphertext~\cite{GPV08} bound to an attribute vector $\y$ satisfying $\langle \x,\y \rangle=0$. In order to improve efficiency,
Xagawa~\cite{Xagawa13} suggested an enhanced variant that employs the primitive matrix $\GG$.  In the below, we will describe the AFV scheme with Xagawa's improvement.  The scheme works with parameters $n,q,\ell,m,\kappa,s,\chi$ and an encoding function $\mathsf{encode}: \{0,1\} \rightarrow \{0,1\}^\kappa$, where $\mathsf{encode}(b) = (b, 0, \ldots, 0) \in \{0,1\}^\kappa$ - the binary vector that has bit $b$ as the first coordinate and $0$ elsewhere.

%To bound error terms during decryption, they use the $r$-ary decomposition of $\x$ for suitably small $r$.  We adopt a variant of the AFV {PE} scheme,  using the binary decomposition for simplicity.

\begin{description}
\item[$\Setup$] \hspace*{-5pt}: Generate $(\A,\T_{\A})\leftarrow \TrapGen(n,q,m)$. Pick $\V\xleftarrow{\$} \Z_q^{n\times \kappa }$ and for each $ i\in [\ell]$, sample $\A_i\xleftarrow{\$}\Z_q^{n \times m}$. Output
    $$
    \pp_{\PE}=(\A,\hspace{2pt}\{\A_i\}_{i\in [\ell]}, \hspace{2pt}\V);\hspace{.7cm}\msk_{\PE}=\T_{\A}.
  $$
\item[$\KeyGen$] \hspace*{-5pt}:  For a predicate vector $\x=(x_1,\ldots,x_{\ell})\in\Z_q^{\ell}$, set $\A_{\x}=\sum\limits_{i=1}^{\ell} \A_i \GG^{-1}(x_i\cdot\GG) \in \mathbb{Z}_q^{n \times m}$ and output $\sk_{\x}= \ZZ$ by running $$\SampleLeft\left(\A,\hspace{2pt}\A_{\x},\hspace{2pt}\T_{\A},\hspace{2pt}\V,\hspace{2pt}s\right).
$$
Another way is to set $\sk_{\x}= \T_{\x}$ by running $$
\SampleBLeft\left(\A,\hspace{2pt}\A_{\x},\hspace{2pt}\T_{\A},\hspace{2pt}s\right).
$$ Note that we then can sample $$\ZZ\leftarrow\SamplePre\left([\A\mid\A_{\x}],\hspace{2pt}\T_{\x},\hspace{2pt}\V,\hspace{2pt}s\right).$$
\item[$\Enc$] \hspace*{-5pt}:  To encrypt a message $M\in\{0,1\}$ under an attribute~$\y=(y_1,\ldots,y_{\ell})\in \Z_q^{\ell}$, choose~$\s\stackrel{\$}{\leftarrow}\Z_q^{n}$, $\e{\hookleftarrow}\chi^{\kappa}$, $\e_1{\hookleftarrow}\chi^m$, $\R_i\stackrel{\$}{\leftarrow}\{-1,1\}^{m\times m}$ for each $i\in [\ell]$, and then output $\ct=(\c,\hspace{2pt}\c_0,\hspace{2pt}\{\c_i\}_{i\in [\ell]})$, where:
\[
\begin{cases}
 \hspace*{1pt}  \c\hspace*{3pt}  = \V^\top{\s}+\e+\mathsf{encode}(M)\cdot\lfloor \frac{q}{2}\rfloor \in\Z_q^{\kappa},\\
  \c_0 =\A^\top \s+\e_1\in\Z_q^{m},\\
 \c_i =\left(\A_i+y_i\cdot\GG\right)^\top \s+\R_i^\top\e_1\in\Z_q^{m}, \forall i\in [\ell]
\end{cases}
\]
  \item[$\Dec$] \hspace*{-5pt}:  Set $\c_{\x}=\sum\limits_{i=1}^{\ell} \left(\GG^{-1}(x_i\cdot\GG)\right)^\top  \c_i \in\Z_q^m $ and  $\mathbf{d}=\c-\ZZ^\top[\c_0 \mid\c_{\x}]\in\Z_q$. If  $\lfloor \frac{2}{q} \cdot \d \rceil = \mathsf{encode}(M')$, for some $M' \in \{0,1\}$, then output~$M'$. Otherwise, output $\bot$.
\end{description}

Agrawal, Freeman and Vaikuntanathan~\cite{AgrawalFV11} showed that, under the $(n,q,\chi)$-\textsf{LWE} assumption, their {PE} scheme  satisfies the  weak attribute-hiding security notion in a selective attribute setting (short as {\sf wAH-sA-CPA}), defined by Katz, Sahai and Waters~\cite{KatzSW08}. Xagawa~\cite{Xagawa13} proved that the same assertion holds for his improved scheme variant.
  We thus have the following theorem.
\begin{theorem}[Adapted from~\cite{AgrawalFV11,Xagawa13}]\label{theorem:PE}
If the $(n,q,\chi)$-{\sf LWE} problem is hard, then the improved AFV PE scheme is {\sf wAH-sA-CPA} secure.
\end{theorem}
%In Section~\ref{section: main scheme}, the scheme will be used as a building block for our lattice-based instantiation of server-aided revocable predicate encryption.

\subsection{The Complete Subtree Method}\label{subsection:RIBE}
%In~\cite{ChenLLWN12}, Chen et al. proposed the first RIBE scheme from lattice assumptions. Their revocation mechanism relies on the Complete Subtree (CS) method of Naor et al.~\cite{NaorNL01}, which was first adapted into the context of RIBE by Boldyreva et al.~\cite{BoldyrevaGK08}. We will briefly recall this method.

%The KGC issues every user a set of private keys, where each key is corresponding to each node on the path from their associated leaf to the root on a complete binary tree. At each time period $\t$, an update key, associated with the set generated by a node selection algorithm, are published to all users. Then, only non-revoked users can obtain the decryption key by combining any one of their node private key with the update key at the \

The complete subtree (CS) method, proposed by Naor, Naor and Lotspiech~\cite{NaorNL01}, has been
widely used in revocation systems. It makes use of a node selection algorithm (called $\mathsf{KUNodes}$).  In the algorithm, we build a complete binary $\BT$ with at least $N$ leaf nodes, where $N$ is the maximum number of users in the system. Each user is corresponding to a leaf node of $\BT$.  We use the following notation: If $\theta$ is a non-leaf node, $\theta_{\ell}$ and $\theta_{r}$
denote the left and right child of~$\theta$, respectively. Whenever $\nu$ is a leaf node, the set $\Path(\nu)$
stands for the collection of nodes on the path from $\theta$ to the root (including $\theta$ and the root). The $\mathsf{KUNodes}$ algorithm takes as input the binary tree~$\BT$, a revocation list $\RL$ and a time period $\t$, and outputs a set of nodes $Y$ which is the smallest subset of nodes that contains  an ancestor of all the leaf nodes corresponding to  non-revoked users.
It is known~\cite{NaorNL01} that the set $Y$ generated by~$\mathsf{KUNodes}(\BT,\RL,\t)$ has a size at most $r\log \frac{N}{r}$, where $r$ is the number of users in $\RL$.
The detailed description
of algorithm $\mathsf{KUNodes}$ is given below.
 and an example is illustrated in Figure~\ref{figcs}.
\vspace{-.5cm}
\begin{center}
\begin{eqnarray*}
&~~& \hspace*{-9pt}\mathsf{KUNodes}(\BT,\RL,\t) \\
&~~&   X,Y\leftarrow \emptyset \\
&~~& \forall (\nu_i,\t_i)\in\RL: \text{ if } \t_i\leq \t, \text{ then add } \Path(\nu)\text{ to } X  \\
&~~& \forall \theta\in X:
  \text{ if } \theta_{\ell}\not\in X, \text{ then add } \theta_{\ell}\text{ to } Y;
   \hspace*{2pt}\text{ if } \theta_{r}\not\in X, \text{ then add } \theta_{r}\text{ to } Y  \\
&~~& \text{If } Y=\emptyset, \text{ then add the root to } Y \\
&~~&  \text{Return } Y
\end{eqnarray*}
\end{center}

\begin{figure}
\begin{center}
\begin{tikzpicture}[,>=stealth',level/.style={sibling distance = 4cm/#1,
  level distance = 0.7cm}]

  \draw(-3.66,-2.1) circle(0.6em);
  \draw(-1.66,-2.1) circle(0.6em);
  \draw(2,-0.7) circle(0.6em);

\node [arn_y] {root}
    child{ node [arn_r] {13}
            child{ node [arn_r] {9}
            	child{ node [arn_r] {1}} %for a named pointer
							child{ node [arn_n] {2}}
            }
            child{ node [arn_r] {10}
							child{ node [arn_r] {3}}
							child{ node [arn_n] {4}}
            }
    }
    child{ node [arn_r] {14}
            child{ node [arn_r] {11}
							child{ node [arn_r] {5}}
							child{ node [arn_r] {6}}
            }
            child{ node [arn_r] {12}
							child{ node [arn_r] {7}}
							child{ node [arn_r] {8}}
            }
		}

;
\end{tikzpicture}
\caption{Assuming that $\RL=\{2,4\}$, it follows that $\{1,3,14\}\leftarrow\mathsf{KUNodes}(\BT,\RL,\t)$. For non-revoked identity ${5}$, node $14\in\mathsf {KUNode}(\BT,\RL)$ is an ancestor of $5$. For identity ${4}\in\RL$, the set $\Path(4)=\{4,10,13,\text{root}\}$ is disjoint with~$\mathsf{KUNodes}(\BT,\RL,\t)$.}
\label{figcs}
\end{center}
\end{figure}
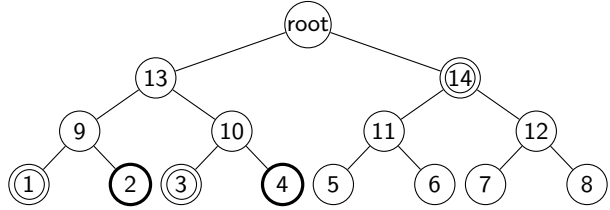

%In the scheme of Section~\ref{section: main scheme}, we will employ the CS method to realize user revocation

\section{Server-Aided Revocable Predicate Encryption}   \label{section: model}
In this section, we describe the rigorous definition and security model of SR-PE, based on the server-aided revocation mechanism advocated by Qin et al.~\cite{SRIBE} and the model of SR-ABE by Cui et al.~\cite{CuiDLQ16}.

The mechanism advocated by Qin et al.~\cite{SRIBE} is depicted in Figure~\ref{fig:SR-PE}. Specifically, when a new recipient joins the system, the KGC issues a private key and a corresponding token both associated with his identity and predicate. The former is given to the recipient and the latter is sent to the server. At each time period, the KGC issues an update key to the server who combines with the stored users' tokens to generate the transformation keys for users. A sender encrypts a message under an attribute and a time period and the ciphertext is sent to the untrusted server. The latter transforms the ciphertext to a partially decrypted ciphertext using a transformation key corresponding to the recipient's identity and the time period bound to the ciphertext. Finally, the recipient recovers the message from the partially decrypted ciphertext using his private key.

\begin{figure}
\begin{center}
\setlength{\unitlength}{1mm}
   \begin{picture}(72, 35)
   \thicklines
%\put(28,40){\includegraphics[width=.7in]{cloud}}
%\put(19,45){\vector(1,0){10}}
%\put(63,45){\vector(-1,0){18}}
%\put(63,45){\vector(0,-1){35}}

\put(20,26.5){\framebox[1.7cm]{KGC}}
\put(52,6.5){\framebox[1.7cm]{Recipient}}
\put(20,6.5){\framebox[1.7cm]{Server}}
\put(-15,6.5){\framebox[1.7cm]{Sender}}

\put(37,28){\line(1,0){23}}
\put(60,28){\vector(0,-1){18}}

\put(15,28){\line(1,0){5}}
\put(15,28){\line(0,-1){12}}
\put(15,16){\vector(1,-1){6}}
\put(28,25.5){\vector(0,-1){15.5}}

\put(2,8){\vector(1,0){18}}
\put(37,8){\vector(1,0){15}}
\put(69,8){\vector(1,0){15}}

\put(41,29){Private key}
\put(5,18){Token}
\put(29,18){Update key}
\put(4,5){\scriptsize{Ciphertext}}
\put(38,5){\scriptsize{Partially}}
\put(38,2.5){\scriptsize{decrypted}}
\put(38,0){\scriptsize{ciphertext}}
\put(70,5){\scriptsize{Message}}
\end{picture}
\caption{Architecture of server-aided revocation mechanism}
\label{fig:SR-PE}
\end{center}
\end{figure}
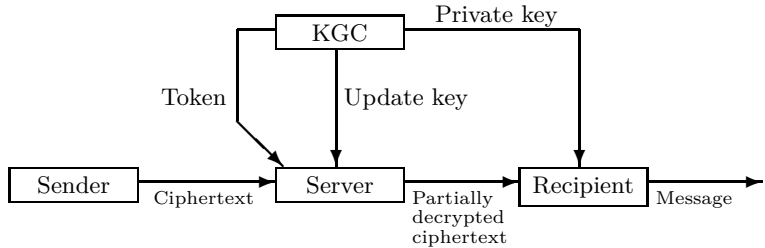

In comparison with Cui et al.'s model of SR-ABE~\cite{CuiDLQ16}, our model offers two crucial differences.
In~\cite{CuiDLQ16}, it is assumed that each user in the system has to maintain a public-secret key-pair (which can possibly be a key-pair for an ordinary PKE scheme). Although this setting can eliminate the need for a secure channel between the KGC and the users (as explained by Cui et al.), we find it somewhat unnatural in the context of identity-based/attribute-based/predicate cryptosystems. (After all, one of the main advantages of these systems over PKE systems is the elimination of users' public keys.) In contrast, our model of SR-PE does not require the users to maintain their own public keys. In the same spirit of IBE systems, we get rid of the notion of users' public keys and we assume a secure channel for transmitting users' private keys.

Another notable difference between our model and~\cite{CuiDLQ16} is due to gap between security notions for ABE and PE systems. Our model preserves the attribute-hiding property of PE systems, which, unlike ABE systems, attributes bound to the ciphertexts are not revealed during decryptions.

A server-aided revocable predicate encryption (SR-PE) scheme involves $4$ parties: KGC, sender, recipient, and untrusted server. It is assumed that the server stores a list of tuples (identity, predicate, token), i.e., $(\id,f,\tau_{\id,f})$. Algorithms among the parties are as follows:

%\begin{description}
%\item[$1.$] (Key distribution: KGC $\longrightarrow$ recipient and server) When a new recipient joins the system, the KGC issues a private key and a corresponding token both associated with his identity and predicate. The former is given to the recipient and the latter is sent to the server.
%\item[$2.$] (Key updates: KGC $\longrightarrow$ server) At each time period, the KGC issues an update key to the server who combines with the stored users' tokens to generate the transformation keys for users.
%\item[$3.$] (Encryption: sender $\longrightarrow$ server) A sender encrypts a message under an attribute and a time period. The ciphertext is sent to the untrusted server.
%\item[$4.$] (Partial decryption: server $\longrightarrow$ recipient) The server transforms the ciphertext to a partially decrypted ciphertext using a transformation key corresponding to the recipient's identity and the time period bound to the ciphertext.
%\item[$5.$] (Decryption: recipient) The recipient recovers the message from the partially decrypted ciphertext using his private key.
%\end{description}

\begin{description}
\item[$\mathsf{Sys}$]\hspace*{-4.5pt}$(1^\lambda)$ is run by the KGC. It takes as input a security parameter $\lambda$ and outputs {the system parameters}~$\params$.\smallskip
  \item[$\Setup$]\hspace*{-4.5pt}$(\params)$ is run by the KGC. It takes as input the system parameters $\params$ and outputs {public parameters}~$\pp$, a master secret key~$\msk$, a revocation list~$\RL$ (initially empty), and a state~$\st$. We assume that~$\pp$ is an implicit input of all other algorithms.\smallskip
\item[$\mathsf{UserKG}$] \hspace*{-4.5pt}$(\msk,\hspace{1pt}\id,\hspace{1pt}f)$ is run by the KGC. It takes as input the master secret key $\msk$ and an identity $\id$ with predicate $f$. It outputs a private key $\sk_{\id,f}$ which is sent to the recipient through a secret channel.\smallskip
  \item[$\Token$]\hspace*{-4.5pt}$(\msk,\hspace{1pt}\id,\hspace{1pt}f,\hspace{1pt}\st)$ is run by the KGC. It takes as input the master secret key $\msk$, an identity $\id$ with a predicate $f$, and state $\st$. It outputs a token $\tau_{\id,f}$ and an updated state $\st$. The token $\tau_{\id,f}$ is sent to the server through a public channel.\smallskip
  \item[$\UpdKG$]\hspace*{-4.5pt}$(\msk,\hspace{1pt}\t,\hspace{1pt}{\RL},\hspace{1pt}\st)$ is run by the KGC. It takes as input the master secret key $\msk$, a time $\t$, the current revocation list $\RL$, and state $\st$. It outputs an update key $\uk_\t$ which is sent to the server through a public channel.\smallskip
  \item[$\TranKG$]\hspace*{-4.5pt}$(\id,\hspace{1pt}\tau_{\id,f},\hspace{1pt}\uk_\t)$ is run by the server. It takes as input an identity with the corresponding token $\tau_{\id,f}$ and an {update} key $\uk_\t$, and outputs a transformation key $\tk_{\id,\t}$ for user $\id$ at the time period $\t$.\smallskip
  \item[$\Enc$]\hspace*{-4.5pt}$(I,\hspace{1pt}\t,\hspace{1pt}M)$ is run by each sender. It takes as input an attribute $I$, a time $\t$, and a message $M$. It outputs a ciphertext $\ct_\t$ which is publicly sent to the server.
  \item[$\Transform$]\hspace*{-4.5pt}$(\ct_\t,\hspace{1pt}\id,\hspace{1pt}\tk_{\id,\hspace{1pt}\t})$ is run by the sever. It takes as input a ciphertext $\ct_\t$, and an identity with the corresponding transform key $\tk_{\id,\t}$. It outputs a partially decrypted ciphertext $\ct'_\id$, which is sent to the recipient with identity $\id$ through a public channel.\smallskip
  \item[$\Dec$]\hspace*{-4.5pt}$(\ct'_\id,\hspace{1pt}\sk_{\id,f})$ is run by each recipient. It takes as input a partially decrypted ciphertext $\ct'_\id$ and a private key $\sk_{\id,f}$. It outputs a message $M$ or symbol~$\bot$.\smallskip
  \item[$\Revoke$]\hspace*{-4.5pt}$(\id,\hspace{1pt}\t,\hspace{1pt}\RL,\hspace{1pt}\st)$ is run by the KGC. It takes as input an identity $\id$ to be revoked, a revocation time $\t$, the current revocation list $\RL$, and a state $\st$. It outputs an updated revocation list $\RL$.
\end{description}

The correctness requirement for an SR-PE scheme states that: For any $\lambda\in \N$, all possible state~$\st$, and any revocation list $\sf RL$, if all parties follow the prescribed algorithms, and if $\id$ is not revoked on a time $\t$, then:
\begin{enumerate}
\item If $f(I)=1$ then $\Dec\left(\ct'_\id,\hspace{1pt}\sk_{\id,f}\right)=M$.\smallskip
\item If $f(I)=0$ then $\Dec\left(\ct'_\id,\hspace{1pt}\sk_{\id,f}\right)=\bot$ with all but negligible probability.
\end{enumerate}

Next, we give the semantic security against selective attributes chosen plaintext attacks for server-aided revocable predicate encryption (short as {\sf SR-sA-CPA}). The selective security means that the adversary needs to be announced the challenge attributes and time period before seeing public parameters. In addition, it is assumed the adversary must commits in advance the set of users to be revoked prior to the challenge time, which is similar to the semi-static query model considered in~\cite{GentryW09,Attrapadung2009}.
\smallskip
\begin{definition}[{\sf SR-sA-CPA} Security]\label{definition:security}
Let~$\mathcal{O}$ be the set of the following oracles: %$\big\{{\bf Token}(\cdot),{\bf UpdKG}(\cdot),{\bf PrivKG}(\cdot),{\bf DecKG}(\cdot,\cdot),{\bf Revoke}(\cdot,\cdot)\big\}$, defined as follows:
\begin{description}
 \item[$-$] ${\UserKG}(\cdot,\cdot)$: On input an identity $\id$ and a predicate $f$, return a private key $\sk_{\id,f}$ by running ${\UserKG}(\msk,\hspace{1pt}\id,\hspace{1pt}f)$.\smallskip
 \item[$-$] ${\Token}(\cdot,\cdot)$: On input an identity $\id$ and a predicate $f$,  return a token~$\tau_{\id,f}$ by running  ${\Token}(\msk,\hspace{1pt}\id,\hspace{1pt}f,\hspace{1pt}\st)$.\smallskip
 \item[$-$]  ${\UpdKG}(\cdot)$: On input a time period $\t$, return an update key $\uk_\t$ by running algorithm~${\UpdKG}(\msk,\hspace{1pt}\t,\hspace{1pt}{\RL},\hspace{1pt}\st)$. If $\t=\t^*$, then ${\RL}^*$ must be a subset of the $\RL$ at $\t^*$.\smallskip
 \item[$-$] ${\Revoke}(\cdot,\cdot)$: On input an identity $\id$ and a time $\t$, return an updated revocation list $\RL$ by running ${\Revoke}(\id,\hspace{1pt}\t,\hspace{1pt}\RL,\hspace{1pt}\st)$. Note that this oracle cannot be queried on time $\t$ if ${\UpdKG}(\cdot)$ has been queried on time $\t$.
\end{description}

An SR-PE scheme is {\sf SR-sA-CPA} secure if any PPT adversary $\Ad$ has negligible advantage in the following experiment:
\vspace{-3pt}
\begin{eqnarray*}
&~&\boxed{{\mathsf{Exp}}^{\text{\sf SR-sA-CPA}}_{\Ad}}(\lambda) \\[2.5pt]
&~~~&  \params\leftarrow{\Sys}(1^{\lambda}); \hspace*{2.5pt} I_0,I_1,\t^*,{\RL}^*\leftarrow \Ad \\
&~~~&(\pp,\msk,\st,\RL)\leftarrow \Setup(\params) \\
&~~~&M_0,M_1\leftarrow \Ad^{\mathcal{O}}(\pp) \\
&~~~&b\stackrel{\$}{\leftarrow}\{0,1\} \\
&~~~&\ct^* \leftarrow \Enc(I_b,\t^*, M_b) \\
&~~~&b'\leftarrow \Ad^{\mathcal{O}}(\ct^*) \\
&~~~&\text{Return } \hspace*{2.5pt} 1 \text{ if } b'=b \text{ and } 0 \text{ otherwise.}
\end{eqnarray*}
Beyond the condition that $M_0,M_1$ have the same length, the following restrictions are made:
\begin{enumerate}[1.]
\item Case 1: if an identity $\id^*$ with predicate $f^*$ satisfying that $f^*(I_0)=1$ or $f^*(I_1)=1$ has be queried to $\UserKG(\cdot,\cdot)$ and ${\Token}(\cdot,\cdot)$, then~$\id^*$ must be included in ${\RL}^*$.\smallskip
\item Case 2: if an identity $\id^*$ with predicate $f^*$ satisfying that $f^*(I_0)=1$ or $f^*(I_1)=1$ is not revoked at $\t^*$, then $(\id^*,f^*)$ should not be queried to the $\UserKG(\cdot,\cdot)$ oracle.
\end{enumerate}

The advantage of $\Ad$ in the experiment is defined as:
\[
  {\textsf{Adv}}^{\text{\sf SR-sA-CPA}}_{\Ad}(\lambda)=\left|\Pr\left[{\mathsf{Exp}}^{\text{\sf SR-sA-CPA}}_{\Ad}(\lambda)=1\right]-\frac{1}{2}\right|.
\]

\end{definition}
\begin{remark}
We can also define an adaptive security notion, where the adversary is not required to specify the challenge attributes $I_0,I_1$ and time period $\t^*$ before seeing the public parameters~$\pp$. Such a notion is obviously stronger than the selective notion defined above.
\end{remark} 

\section{An SR-PE Scheme from Lattices} \label{section: main scheme}
Our lattice-based SR-PE scheme can be seen as a combination of two AFV PE instances~\cite{AgrawalFV11}, one IBE instance~\cite{ABB10} and the CS method~\cite{NaorNL01}. Each recipient's identity $\id$ corresponds to a matrix $\D_{\id}$ determined by the IBE system.
The KGC generates the private key for the first PE instance with a hierarchical level for~$\D_{\id}$, and issues the token by embedding $\D_\id$ into the second PE scheme as well as using nodes in $\Path(\id)$. At each time period~$\t$, the KGC computes an update key using nodes in $\KUNodes(\BT,\RL,\t)$. Recall that  token and update key are both sent to the sever, who makes use of the intersected node in $\Path(\id)\cap\KUNodes(\BT,\RL,\t)$ to obtain a transformation key. Then, a ciphertext in our scheme is a combination of two PE ciphertexts and an extra component bound to $\t$, where all components have the same randomness (i.e., vector $\s$).
If recipient $\id$ is not revoked at time period $\t$, e.g., $\id\not\in\RL$, then the server can partially decrypt the ciphertext, via the decryption algorithm of the second PE instance.
Finally, the partially decrypted ciphertext contains a proper ciphertext for the first PE system and an additional component bound to matrix $\D_\id$ (all with randomness $\s$) so that it can be fully decrypted using the private key of $\id$ (obtained from the first PE instance and specified by $\D_{\id}$).

In the following, we will formally describe the scheme.

\begin{description}
\item[$\Sys$]\hspace*{-4.5pt}$(1^{\lambda})$: On input security parameter $\lambda$, the KGC performs the following steps:

\smallskip
\begin{enumerate}[leftmargin=*]
\item Set $n=O\left(\lambda\right)$. Choose $N=\mathsf{poly}(\lambda)$ as the maximal number of users the system will support, and arbitrary $\ell$ be the length of predicate and attribute vectors. Choose $\kappa=\omega(\log \lambda)$ as a dimension parameter.\smallskip
\item Let $q=\widetilde{O}\left(\ell^2 n^4\right)$ be a prime power, and set $m= 2n\lceil \log q \rceil$. Note that parameters $n,q,m$ specify the primitive matrix $\mathbf{G}$ (see Section~\ref{subsection:lattice-background}).\smallskip
\item Choose a Gaussian parameter $s=\widetilde{O}\left(\sqrt{m} \right)$.\smallskip
\item Set $B=\widetilde{O}\left(\sqrt{n} \right)$ and let $\chi$ be a $B$-bounded distribution.\smallskip
\item Select an efficient full-rank difference map $\H: \Z_q^n \rightarrow \Z_q^{n \times n}$.\smallskip
\item Let the identity space be $\mathcal{I}\subseteq\Z_q^n$, the time space be $\mathcal{T}\subseteq\Z_q^n$, the message space be $\mathcal{M}=\{0,1\}$, the predicate space be $\mathbb{P}=\{f_{\x}\hspace*{1.8pt}\big| \hspace*{1.8pt}\x\in\Z_q^{\ell}\}$ and the attribute space be $\mathbb{A}=\Z_q^{\ell}$ (see Section~\ref{subsection:PE}).
\item Define the encoding function $\mathsf{encode}$ (see Section~\ref{subsection:PE} ).\smallskip
\item Output $\params=\left(n,\hspace{1pt}N,\hspace{1pt}\ell,\hspace{1pt}\kappa,\hspace{1pt}q,\hspace{1pt}m,\hspace{1pt}s,\hspace{1pt}B,\hspace{1pt}\chi,\hspace{1pt}\H,\hspace{1pt}\mathcal{I},\hspace{1pt}\mathcal{T},\hspace{1pt}\mathcal{M},\hspace{1pt}\mathbb{P},\hspace{1pt}\mathbb{A},\hspace{1pt}\mathsf{encode}\right)$.
\end{enumerate}

\medskip
\item[$\Setup$]\hspace*{-4.5pt}$(\params)$: On input the system parameters~$\params$, the KGC performs the following steps:
\begin{enumerate}[leftmargin=*]
  \item Generate independent pairs $(\A,\T_{\A})$ and $(\B,\T_{\B})$ using $\TrapGen(n,q,m)$.\smallskip
  \item Select $\V\stackrel{\$}{\leftarrow}\Z^{n\times  \kappa}_q$ and $\C,\hspace{2pt}\D,\hspace{2pt}\A_i,\hspace{2pt}\B_i\stackrel{\$}{\leftarrow}\Z_q^{n\times m}$ for each $i\in [\ell]$.\smallskip
  \item Initialize the revocation list $\RL=\emptyset$. Obtain a binary tree $\BT$ with at least $N$ leaf nodes and set the state~$\st=\BT$.
  \item Set $\pp=\left(\A,\hspace{2pt}\B,\hspace{2pt}\C,\hspace{2pt}\D,\hspace{2pt}\{\A_i\}_{i\in [\ell]},\hspace{2pt}\{\B_i\}_{i\in [\ell]},\hspace{2pt}\V\right)$ and $\msk=\left(\T_{\A},\hspace{2pt}\T_{\B}\right)$.\smallskip
  \item Output $\left(\pp,\hspace{2pt}\msk,\hspace{2pt}\RL,\hspace{2pt}\st\right)$.\smallskip
\end{enumerate}

\medskip
\item[$\UserKG$]\hspace*{-4.5pt}$(\msk,\hspace{2pt}\id,\hspace{2pt}\x)$: On input the master secret key $\msk$ and an identity $\id\in \mathcal{I}$ with predicate vector $\x=(x_1,\ldots,x_\ell)\in\Z_q^{\ell}$, the KGC performs the following steps:
\smallskip
\begin{enumerate}[leftmargin=*]
\item Set $\B_{\x}=\sum\limits_{i=1}^\ell \B_i\GG^{-1}(x_i\cdot\GG)\in\Z_q^{n\times m}$ and~$\D_{\id}=\D+\H(\id)\GG\in\Z_q^{n\times m}$.\smallskip
  \item Sample  $\ZZ\leftarrow \SampleLeft\left(\B,\hspace{2pt}[\B_{\x}\mid\D_{\id}],\hspace{2pt}\T_{\B},\hspace{2pt}\V,\hspace{2pt}s\right)$. Note that $\ZZ\in\Z^{3m\times \kappa}$ and $[\B\mid\B_{\x}\mid\D_\id]\cdot\ZZ=\V$.\smallskip
   \item Output $\sk_{\id,\x}=\ZZ$.\smallskip
\end{enumerate}

\medskip
\item[$\Token$]\hspace*{-4.5pt}$(\msk,\hspace{2pt}\id,\hspace{2pt}\x,\hspace{2pt}\st)$: On input the master secret key $\msk$, an identity $\id\in \mathcal{I}$ with predicate vector $\x=(x_1,\ldots,x_\ell)\in\Z_q^{\ell}$, and state $\st$, the KGC performs the following steps:
\smallskip
\begin{enumerate}[leftmargin=*]
\item Compute $\A_{\x}=\sum\limits_{i=1}^\ell \A_i\GG^{-1}(x_i\cdot\GG)\in\Z_q^{n\times m}$.\smallskip
\item For each $\theta\in \Path({\id})$, if~$\U_{\theta}$ is undefined, then pick $\U_{\theta}\stackrel{\$}{\leftarrow} \Z_q^{n\times m}$ and store it on $\theta$; Sample
  $
  \ZZ_{1,\theta}\leftarrow \SampleLeft\left(\A,\hspace{2pt}\A_{\x},\hspace{2pt}\T_{\A},\hspace{2pt}\D_{\id}-\U_{\theta},\hspace{2pt}s\right)
  $. Note that $\ZZ_{1,\theta}\in\Z^{2m\times m}$ and $[\A\mid\A_{\x}]\cdot\ZZ_{1,\theta}=\D_{\id}-\U_{\theta}$.
\item Output the updated state $\st$ and $\tau_{\id,\x}=\{\theta,\hspace{2pt}\ZZ_{1,\theta}\}_{\theta\in \Path({\id})}$.\smallskip
\end{enumerate}

\medskip
\item[$\UpdKG$]\hspace*{-4.5pt}$(\msk,\hspace{2pt}\t,\hspace{2pt}\st,\hspace{2pt}\RL)$: On input the master secret key $\msk$, a time $\t\in \mathcal{T}$, the revocation list $\RL$ and state $\st$, the KGC performs the following steps:
\smallskip
\begin{enumerate}[leftmargin=*]
\item Compute $\C_{\t}=\C+\H(\t)\GG\in\Z_q^{n\times m}$.\smallskip
\item For each $\theta\in \KUNodes(\BT,\RL,\t)$, retrieve $\U_{\theta}$ (which is always pre-defined in algorithm {$\mathsf{Token}$}), and sample $\ZZ_{2,\theta}\leftarrow \SampleLeft\left(\A,\hspace{2pt}\C_{\t},\hspace{2pt}\T_{\A},\hspace{2pt}\U_{\theta},\hspace{2pt}s\right)$.
 Note that $\ZZ_{2,\theta}\in\Z^{2m\times m}$ and $[\A\mid\C_{\t}]\cdot\ZZ_{2,\theta}=\U_{\theta}$.\smallskip
\item Output $\uk_{\t}=\{\theta,\hspace{2pt}\ZZ_{2,\theta}\}_{\theta\in \KUNodes(\BT,\RL,\t)}$.\smallskip
\end{enumerate}

\medskip
\item[$\TranKG$]\hspace*{-4.5pt}$(\id,\hspace{2pt}\tau_{\id,\x},\hspace{2pt}\uk_\t)$: On input an identity $\id$ with token $\tau_{\id,\x}=\{\theta,\hspace{2pt}\ZZ_{1,\theta}\}_{\theta\in I}$ and an update key $\uk_{\t}=\{\theta,\hspace{2pt}\ZZ_{2,\theta}\}_{\theta\in J}$ for some set of nodes~$I,J$, the server performs the following steps:
\smallskip
\begin{enumerate}[leftmargin=*]
 \item If $I\cap J=\emptyset$, output $\bot$.\smallskip
  \item Otherwise, choose $\theta\in I\cap J$ and output $\tk_{\id,\t}=(\ZZ_{1,\theta},\hspace{2pt}\ZZ_{2,\theta})$. Note that $[\A\mid\A_{\x}]\cdot\ZZ_{1,\theta}+ [\A\mid\C_{\t}]\cdot\ZZ_{2,\theta}=\D_\id$.\smallskip
\end{enumerate}

\medskip

\item[$\Enc$]\hspace*{-4.5pt}$(\y,\hspace{2pt}\t,\hspace{2pt}M)$: On input an attribute vector $\y=(y_1,\ldots,y_\ell)\in\Z_q^\ell$, a time $\t\in\mathcal{T}$ and a message $M\in\mathcal{M}$, the sender performs the following steps:
\smallskip
 \begin{enumerate}[leftmargin=*]
 \item Sample $\s\stackrel{\$}{\leftarrow}\Z_q^{n}$, $\e_1,\e_2{\hookleftarrow}\chi^{m}$ and $\e{\hookleftarrow}\chi^{\kappa}$.\smallskip
 \item Choose $\bar{\R},\hspace{2pt}\S_i,\hspace{2pt}\R_i\stackrel{\$}{\leftarrow}\{-1,1\}^{m\times m}$ for each $i\in [\ell]$.\smallskip
\item Output $\ct_\t=(\c,\hspace{2pt}\c_{1},\hspace{2pt}\{\c_{1,i}\}_{i\in[\ell]},\hspace{2pt}\c_{1,0},\hspace{2pt}\c_{2},\hspace{2pt}\{\c_{2,i}\}_{i\in[\ell]})$ where:
\[
\begin{cases}
\c = \V^\top{\s}+\e+\mathsf{encode}({M})\cdot\lfloor \frac{q}{2}\rfloor \in\Z_q^{\kappa}, \\[1pt]
\c_1 = \A^\top\s+\e_1 \in \Z_q^m, \\[1pt]
\c_{1,i} = (\A_i+y_i\cdot\GG)^\top\s+\R_i^\top\e_1 \in \Z_q^{m}, \hspace*{6.6pt}\forall \hspace*{1pt} i\in[\ell]\\[1pt]
\c_{1,0} = \C_\t^\top\s+\bar{\R}^\top\e_1 \in \Z_q^{m},\\[1pt]
\c_2 = \B^\top\s+\e_2\in \Z_q^m, \\[1pt]
\c_{2,i} = (\B_i+y_i\cdot\GG)^\top\s+\S_i^\top\e_2 \in \Z_q^{m}, \hspace*{6.6pt}\forall \hspace*{1pt} i\in[\ell].
\end{cases}
\]
\end{enumerate}
\medskip
\item[$\Transform$]\hspace*{-4.5pt}$(\ct_\t,\hspace{2pt}\id,\hspace{2pt}\tk_{\id,\t})$: On input~$\ct_\t=(\c,\hspace{2pt}\c_{1},\hspace{2pt}\{\c_{1,i}\}_{i\in[\ell]},\hspace{2pt}\c_{1,0},\hspace{2pt}\c_{2},\hspace{2pt}\{\c_{2,i}\}_{i\in[\ell]})$ and an identity $\id$ with transformation key $\tk_{\id,\t}=(\ZZ_1,\ZZ_2)$, the server performs the following steps:
\smallskip
\begin{enumerate}[leftmargin=*]
\item Compute $\c_{1,\x}=\sum\limits_{i=1}^{\ell} \left(\GG^{-1}(x_i\cdot\GG)\right)^\top  \c_{1,i}\in\Z_q^m $.\smallskip
\item  Compute~$\bar{\c}=\ZZ_1^\top[\c_{1}\mid\c_{1,\x}]+\ZZ_2^\top[\c_1\mid\c_{1,0}]\in\Z_q^\kappa$.\smallskip
 \item Output $\ct'_\id=(\c,\c_{2},\{\c_{2,i}\}_{i\in[\ell]},\bar{\c})$.\smallskip
\end{enumerate}

\vspace{.3cm}
\item[$\Dec$]\hspace*{-4.5pt}$(\ct'_\id,\hspace{2pt}\sk_{\id,\hspace{2pt}\x})$: On input~$\ct'_\id=(\c,\hspace{2pt}\c_{2},\hspace{2pt}\{\c_{2,i}\}_{i\in[\ell]},\hspace{2pt}\bar{\c})$ and a private key~$\sk_{\id,\x}=\ZZ$, the recipient performs the following steps:
\smallskip
\begin{enumerate}[leftmargin=*]
\item Compute $\c_{2,\x}=\sum\limits_{i=1}^{\ell} \left(\GG^{-1}(x_i\cdot\GG)\right)^\top  \c_{2,i}\in\Z_q^m $.\smallskip
\item Compute $\mathbf{d}=\c-\ZZ^\top[\c_2\mid\c_{2,\x}\mid\bar{\c}] \in\Z_q^\kappa$.\smallskip
\item If  $\lfloor \frac{2}{q} \cdot \d \rceil = \mathsf{encode}(M')$, for some $M' \in \{0,1\}$, then output~$M'$. Otherwise, output $\bot$.\smallskip
\end{enumerate}

\vspace{.3cm}
\item[$\Revoke$]\hspace*{-4.5pt}$(\id,\hspace{2pt}\t,\hspace{2pt}\RL,\hspace{2pt}\st)$: On input an identity $\id$, a time $\t$, the revocation list $\RL$ and state $\st=\BT$, the KGC adds $(\id,\t)$ to $\RL$ for all nodes associated with identity $\id$ and returns $\RL$.
\end{description}

\section{Analysis}\label{section: analysis}
\subsection{Correctness and Efficiency}
%We first analyze the correct and efficiency of the proposed scheme.

\noindent
{\bf Correctness. } We will demonstrate that the scheme satisfies the correctness requirement with all but negligible probability. We proceed as in~\cite{AgrawalFV11,Xagawa13,GayMW15,LingNWZ17}.

Suppose that $\ct_\t=(\c,\hspace{2pt}\c_{1},\hspace{2pt}\{\c_{1,i}\}_{i\in[\ell]},\hspace{2pt}\c_{1,0},\hspace{2pt}\c_{2},\hspace{2pt}\{\c_{2,i}\}_{i\in[\ell]})$ is an honestly computed ciphertext of message $M \in \mathcal{M}$, with respect to some $\y \in \mathbb{A}$. Let $\tk_{\id,\t}=(\ZZ_1,\hspace{2pt} \ZZ_{2})$ be a correctly generated transformation  key, where $\id$ is not revoked at time $\t$.  Then we have:
 $$[\A\mid\A_{\x}]\cdot\ZZ_{1}+ [\A\mid\C_{\t}]\cdot\ZZ_{2}=\D_\id.$$
We also observe that the following two equations hold:
\begin{eqnarray*}
\c_{1,\x} = \sum\limits_{i=1}^{\ell} \left(\GG^{-1}(x_i\cdot\GG)\right)^\top  \c_{1,i} &=& \left(\A_{\x}+\langle \x,\y \rangle\cdot \GG\right)^ \top\s+(\R_\x)^{\top}  {\e_1},\\
\c_{2,\x} =\sum\limits_{i=1}^{\ell} \left(\GG^{-1}(x_i\cdot\GG)\right)^\top  \c_{2,i} &=& \left(\B_{\x}+\langle \x,\y \rangle\cdot \GG\right)^ \top\s+(\S_\x)^{\top}  {\e_2}.
\end{eqnarray*}
where $\R_\x= \sum\limits_{i=1}^{\ell} \R_i\GG^{-1}(x_i\cdot\GG)$ and $\S_\x= \sum\limits_{i=1}^{\ell} \S_i\GG^{-1}(x_i\cdot\GG)$.
\noindent
We now consider two cases:
\begin{enumerate}
\item Case $1$: Suppose that $\langle \x,\y \rangle=0$. In this case, we have:
$\c_{1,\x}=(\A_{\x})^\top\s+(\R_\x) ^{\top}{\e_1}$ and $\c_{2,\x}=(\B_{\x})^\top\s+(\S_\x) ^{\top}{\e_2}$.
Then in $\Transform$ algorithm, the following holds:

\begin{align*}
\bar{\c} & =\ZZ_1^\top[\c_{1}\mid\c_{1,\x}]+\ZZ_2^\top[\c_1\mid\c_{1,0}] \\
     & = \ZZ_1^\top\left(\left[\A\mid\A_{\x}\right]^\top\s+\left[\begin{array}{c}  \e_1 \\ (\R_\x)^{\top}{\e_1} \end{array}\right]\right)+\ZZ_{2}^\top\left(\left[\A\mid\C_{\t}\right]^\top\s+\left[\begin{array}{c}  \e_1 \\ \bar{\R}^\top\e_1 \end{array}\right]\right)\\[5pt]
  & = \D_{\id}^\top{\s}+\underbrace{\ZZ_1^\top\left[\begin{array}{c}  \e_1 \\ (\R_\x)^{\top}{\e_1} \end{array}\right]+\ZZ_{2}^\top\left[\begin{array}{c}  \e_1 \\ \S^\top\e_1 \end{array}\right]}_{\mathsf{error}'}
\end{align*}

and in $\Dec$ algorithm, the following holds:
\begin{align*}
\mathbf{d}&= \c-\ZZ^\top[\c_2\mid\c_{2,\x}\mid\bar{\c}] \\
          &=\V^\top\s+\e+\lfloor \frac{q}{2}\rfloor  \cdot \mathsf{encode}(M)-\ZZ^\top \left( [\B\mid \B_\x\mid \D_{\id}]^\top\s+ \left[\begin{array}{c}  \e_2 \\ (\S_{\x})^\top\e_2 \\ \mathsf{error'} \end{array}\right]  \right)\\
          &= \lfloor \frac{q}{2}\rfloor  \cdot \mathsf{encode}(M)+ \underbrace{\e-\ZZ^\top\left[\begin{array}{c}  \e_2 \\ (\S_{\x})^\top\e_2 \\ \mathsf{error'} \end{array}\right]}_{\mathsf{error}}.
\end{align*}
As in \cite{ABB10,AgrawalFV11,Xagawa13,GayMW15,LingNWZ17}, the above error term  can be showed to be bounded by $s\ell m^{2}B\cdot\omega(\log n)=\widetilde{O}(\ell^2 n^{3})$, with all but negligible probability.
In order for the decryption algorithm to recover $\mathsf{encode}(M)$, and subsequently the plaintext $M$, it is required that the error term  is bounded by $q/5$, i.e., $||\mathsf{error}||_{\infty}<q/5$. This is guaranteed by our setting of modulus $q$, i.e., $q= \widetilde{O}\left( \ell^2  n^4 \right)$.

%\begin{enumerate}
\item Case $2$: Suppose that $\langle \x,\y \rangle \neq 0$. In this case, we have:
\begin{eqnarray*}
\c_{2,\x}=\big(\A_{\x}+\underbrace{\langle \x,\y \rangle}_{\neq 0}\cdot \GG\big)^ \top\s+(\S_\x) ^{\top}{\e_2}.
\end{eqnarray*}

Then in $\Dec$ algorithm, $\mathbf{d}= \c-\ZZ^\top[\c_2\mid\c_{2,\x}\mid\bar{\c}]$ contains the following term:
\[
\ZZ^\top [\0\mid \langle \x,\y \rangle \cdot\GG\mid \0]^\top\s\in\Z_q^\kappa.
\]
which can be written as $\langle \x,\y \rangle \cdot  (\GG\ZZ^2)^\top\s$, where $\ZZ^2\in\Z^{m\times \kappa}$ is the middle part of matrix $\ZZ$. By Lemma \ref{lemma:distribution}, we have that the distribution of $\GG\ZZ^2 \in \mathbb{Z}_q^{n \times \kappa}$ is statistically close to uniform. This implies that, vector $\d \in \mathbb{Z}_q^\kappa$ in $\Dec$ algorithm, is indistinguishable from uniform. As a result, the probability that the last $\kappa-1$ coordinates of vector $\lfloor \frac{2}{q} \cdot \d \rceil$ are all $0$ is at most $2^{-(\kappa-1)}=2^{-\omega(\log \lambda)}$, which is negligible in $\lambda$. In other words, except for negligible probability, the decryption algorithm outputs $\bot$ since it does not obtain a proper encoding $\mathsf{encode}(M') \in \{0,1\}^\kappa$, for  $M' \in \{0,1\}$.

\end{enumerate}

\medskip

\noindent
{\bf Efficiency. }  The efficiency aspect of our {SR-PE} scheme is as follows:
\begin{itemize}[leftmargin=*]
 \item[$\diamond$] The bit-size of public parameters $\pp$ is $((2\ell+4)nm+n\kappa)\log q=\widetilde{O}(\ell) \cdot \widetilde{O}\left( \lambda^2 \right)$.\smallskip
  \item[$\diamond$] The private key $\sk_{\id,\x}$ has bit-size $\widetilde{O}(\lambda)$.\smallskip
  \item[$\diamond$] The token $\tau_{\id,\x}$ has bit-size $ O(\log N)\cdot\widetilde{O}\left( \lambda^2 \right)$.\smallskip
 \item[$\diamond$] The update key $\uk_\t$ has bit-size $  O\big(r\log \frac{N}{r}\big)\cdot\widetilde{O}\left( \lambda^2 \right)$.\smallskip
 \item[$\diamond$] The bit-size of  the ciphertext $\ct_\t$  is $\widetilde{O}(\ell\lambda)$.\smallskip
 \item[$\diamond$] The bit-size of  the partially decrypted ciphertext $\ct'_\id$ is $\widetilde{O}(\ell\lambda)$.
\end{itemize}

%The efficiency of our scheme is comparable to that of the pairing-based {RPE} scheme from~\cite{NietoMS12,NMS12-eprint}, in the following sense: the size of public parameters is $O(N)$; the size of the private key is $O(\log N)$, and the ciphertext has size  $O\left(r\log \frac{N}{r}\right)$ which is ranged between $O(1)$ (when no key is revoked) and $O\left(\frac{N}{2}\right)$ (in the worst case when every second key is revoked).
%
%In Section \ref{section: collusion}, we will discuss a variant of our scheme in the random oracle model, which has shorter public parameters.

\subsection{Security} \label{section: security}

In the following theorem, we prove that our scheme in Section \ref{section: main scheme} is {\sf SR-sA-CPA} secure in the standard model, under the $\mathsf{LWE}$ assumption.

\begin{theorem}\label{theorem:SRPE}
Our SR-PE scheme  satisfies the  {\sf SR-sA-CPA} security defined in Definition \ref{definition:security}, assuming hardness of the  $\left({n,q,\chi}\right)$-$\mathsf{LWE}$ problem.
\end{theorem}

\begin{proof}
We will demonstrate that if there is a PPT adversary $\mathcal{A}$ succeeding in breaking the {\sf SR-sA-CPA} security of our SR-PE scheme, then we can use it to construct a PPT algorithm $\mathcal{S}$ breaking the {\sf wAH-sA-CPA} security of the AFV PE scheme. Then the theorem follows from the fact that the  building block is secure under the $(n,q,\chi)$-$\mathsf{LWE}$ assumption (see Theorem \ref{theorem:PE}).

%From Theorem \ref{theorem:HIBE} and Theorem \ref{theorem:RIBE}, the underlying RIBE and HIBE are both secure based on the hardness of the LWE problem, so our SR-IBE is also secure provided that the LWE assumption holds.

Let $\y_0,\y_1$ be the challenge attribute vectors, $\t^*$ be the challenge time and $\RL^*$ be the set of revoked users at $\t^*$. We assume that, without loss of generality, the adversary will make token or private key queries on identities whose predicates are satisfied by $\y_0$ or $\y_1$.  We consider two types of adversaries as follows.
\begin{description}
\item[Type I Adversary:] It is assumed that, every identity $\id^*$ whose predicate vector $\x^*$ satisfies that $\langle \x^*,\y_0 \rangle =0$ or $\langle \x^*,\y_1 \rangle =0$, must be included in $\RL^*$. In this case, the adversary is allowed to issue a query to oracle $\UserKG(\cdot,\cdot)$ on such a pair $(\id^*,\x^*)$. \smallskip
\item[Type \rom{2} Adversary:] It is assumed that there exists an $\id^*\not\in\RL^*$ whose predicate vector $\x^*$ satisfies that $\langle \x^*,\y_0 \rangle =0$ or $\langle \x^*,\y_1 \rangle =0$. In this case, $\id^*$ may be not revoked at $\t^*$ and the adversary never issues a query to oracle $\UserKG(\cdot,\cdot)$ on $(\id^*,\x^*)$.
%In this case, an $\id^*\not\in\RL^*$ whose predicate vector $\x^*$ satisfies that $\langle \x^*,\y_0 \rangle =0$ or $\langle {\x}^*,{\y} _1 \rangle =0$ may be non-revoked at $\t^*$.
\end{description}

\noindent
Algorithm $\mathcal{S}$ begins by randomly guessing the type of adversaries it is going to deal with. Let $Q$ be the number of users in $\RL^*$.
We separately describe algorithm $\mathcal{S}$'s progress for the two types of adversaries.

\begin{lemma}\label{lemma: security of type 1}
If there is a PPT Type \rom{1} adversary $\Ad$ breaking the {\sf SR-sA-CPA} security of our SR-PE scheme  with advantage $\epsilon$, then there is a PPT algorithm $\Sd$ breaking the {\sf wAH-sA-CPA} security of the AFV PE scheme with advantage $\epsilon/Q$.
\end{lemma}
\begin{proof}
Recall that if an identity $\id$ has the predicate vector $\x$ satisfied by the challenge attributes $\y_0$ or $\y_1$, it must be include in $\RL^*$. The simulator $\Sd$ randomly choose $j^*\xleftarrow{\$} [Q]$, at which such an identity appears. Let $\id^*$ be the $j^*$-th user in $\RL^*$ and $\x^*$ be the corresponding predicate vector.

Let $\Bd$ be the challenger in the {\sf wAH-sA-CPA} security game for the AFV PE scheme. Algorithm $\Sd$ interacts with $\Ad$ and $\Bd$ as follows.

\begin{description}

\item[Initial:] $\Sd$ runs algorithm $\mathsf{Sys}\left(1^{\lambda}\right)$ to output $\params$. Then $\Ad$ announces to $\Sd$ the target attribute vectors $\y_{0},\y_{1}$, time $\t^*$ and revocation list $\RL^*$. Algorithm  $\Sd$ forwards $\y_{0},\y_{1}$ to $\Bd$.\smallskip
\item[Setup:] $\Sd$ sets an empty revocation list $\RL$ and a binary tree $\BT$ as the sate $\st$. Then $\Sd$ prepares the public parameters as follows:
\smallskip

\begin{enumerate}[leftmargin=*]
\item Get $\pp_{\PE}=(\A,\hspace{2pt}\{\A_i\}_{i\in [\ell]},\hspace{2pt}\V)$ from~$\Bd$, where $\A,\A_i$ $\in\Z_q^{n\times m},\V\in\Z_q^{n\times \kappa}$.\smallskip
\item Generate $(\B,\T_{\B})$ by running $\TrapGen(n,q,m)$. Pick $\V \stackrel{\$}{\leftarrow}\Z^{n\times \kappa}_q$ and $\D,\B_i\stackrel{\$}{\leftarrow} \Z_q^{n\times m}$ for each $i\in [\ell]$.\smallskip
\item Select $\bar{\R}\stackrel{\$}{\leftarrow} \{-1,1\}^{m\times m}$ and set $\C=\A\bar{\R}-\H(\t^*)\GG\in\Z_q^{n\times m}$.\smallskip
\item Let $\pp=\left(\A,\hspace{2pt}\B,\hspace{2pt}\C,\hspace{2pt}\D,\hspace{2pt}\{\A_i\}_{i\in [\ell]},\hspace{2pt}\{\B_i\}_{i\in [\ell]},\hspace{2pt}\V\right)$, and send $\pp$ to the adversary $\Ad$. Note that the distribution of $\pp$ is exactly the one expected by $\Ad$.\smallskip
\end{enumerate}
\smallskip
\item[Private Key Oracle:]  When $\Ad$ issues a private key query, $\Sd$ performs the same as in the real scheme since it knows the master secret key part $\T_{\B}$.\smallskip

\item[Token and Update Key Oracles:]  The simulator first defines $\U_{\theta}$ for each $\theta\in\BT$ as follows:
\smallskip
\begin{enumerate}[leftmargin=*]
\item If $\theta\in\Path(\id^*)$, pick $\ZZ_{1,\theta}\hookleftarrow\mathcal{D}_{\Z^{2m\times m},s}$ and set~$\U_{\theta}=\D_{\id^*}-[\A|\A_{\x^*}]\cdot\ZZ_{1,\theta}$.\smallskip
\item If $\theta\not\in\Path(\id^*)$, pick $\ZZ_{2,\theta}\hookleftarrow\mathcal{D}_{\Z^{2m\times m},s}$ and set $\U_{\theta}=[\A|\C_{\t^*}]\cdot\ZZ_{2,\theta}$. \smallskip
\end{enumerate}

\smallskip
If $\Ad$ queries a token for $(\id,\x)$ such that $\langle \x,\y_0 \rangle \neq 0$ and $\langle \x,\y_1 \rangle \neq 0$, algorithm $\Sd$ forwards $\x$ to $\Bd$.  Receiving a PE private key $\T_{\x}$ from $\Bd$, algorithm $\Sd$ performs as in the real scheme except that algorithm $$\mathsf{Sampre}([\A\mid\A_{\x}],\hspace{2pt}\T_{\x},\hspace{2pt}\D_{\id}-\U_{\theta},\hspace{2pt}s)$$ replaces algorithm $\mathsf{SampleLeft}$.
\smallskip
If $\Ad$ queries a token for  $(\id,\x)\neq (\id^*,\x^*)$ together with $\langle \x,\y_0 \rangle = 0$ or $\langle \x,\y_1 \rangle = 0$, the simulator returns $\bot$. For the query on $(\id^*,\x^*)$, it returns~$\{\theta,\ZZ_{1,\theta}\}_{\theta\in\Path({\id^*})}$ as defined above. Since the specific $id^*$ is unknown in $\Ad$'s view, $\Sd$ can simulate successfully with probability at least $1/Q$.

\smallskip
For update key of $\t\neq\t^*$, note~$\C_{\t}=\C+\H(\t)\GG=\A\bar{\R}+(\H(\t)-\H(\t^*))\GG$. Algorithm $\Sd$ can compute~$\uk_\t$ as in the real scheme except that algorithm $$\mathsf{SampRight}(\A,\hspace{2pt}\bar{\R},\hspace{2pt}(\H(\t)-\H(\t^*)\GG),\hspace{2pt}\T_{\GG},\U_{\theta},\hspace{2pt}s)$$ replaces algorithm $\mathsf{SampleLeft}$. For the challenge time period $\t^*$, the simulator $\Sd$ returns $\{\theta,\ZZ_{2,\theta}\}_{\theta\in\KUNodes(\BT,\RL,\t^*)} $ as defined above since $\KUNodes(\BT,\RL,\t^*)$ is disjoint with $\Path({\id^*})$.

\smallskip
Next, we observe that, the columns of these matrices are sampled via algorithm $\SampleLeft$ in the real scheme, while they are either sampled via algorithm $\SampleRight$, $\SamplePre$ or sampled
from $\mathcal{D}_{\Z^m,s}$ in the simulation. The properties of these sampling algorithms (see Section~\ref{section: background})
will guarantee that the two distributions are statistically indistinguishable.
\smallskip
\item[Challenge:] $\Ad$ gives two messages $M_0,M_1\in\mathcal{M}$ to $\Sd$ who prepares the challenge ciphertext as follows:
\begin{enumerate}[leftmargin=*]
\item Sample $\s\stackrel{\$}{\leftarrow}\Z_q^{n}$ and $\e_2{\hookleftarrow}\chi^{m}$.  Choose $\S_i\xleftarrow{\$}\{-1,1\}^{m\times m}$ for each $i\in [\ell]$.\smallskip
  \item Pick $d\stackrel{\$}{\leftarrow}\{0,1\}$. Set
  $M'_0=M_d,\hspace*{5pt} M'_1=M_{1\oplus d}$, where $\oplus$ denotes the addition modulus $2$.

Forward $M'_0,M'_1$ as two challenge messages to the PE challenger $\Bd$. The latter chooses $c\stackrel{\$}{\leftarrow}\{0,1\}$ and returns a ciphertext $(\c',\hspace{2pt}\c'_0,\hspace{2pt}\{\c'_i\}_{i\in[\ell]})$ as a PE encryption of $M'_c$ under attribute vector $\y_c$. \smallskip
 \item Output $\ct^*=(\c^*,\hspace{2pt}\c^*_{1},\hspace{2pt}\{\c^*_{1,i}\}_{i\in[\ell]},\hspace{2pt}\c^*_{1,0},\hspace{2pt}\c^*_{2},\hspace{2pt}\{\c^*_{2,i}\}_{i\in[\ell]})$ as an SR-PE encryption of $M_d$ under $\y_d,\t^*$, where:
\[
\begin{cases}
\c^* =\c' \in\Z_q^{\kappa}, \\
\c^*_1 = \c'_0 \in \Z_q^m, \\
\c^*_{1,i} = \c'_i \in \Z_q^{m}, \hspace*{6.6pt}\forall \hspace*{1pt} i\in[\ell]\\[1pt]
\c^*_{1,0} = \bar{\R}^\top\c'_0 \in \Z_q^{m},\\
\c^*_2 = \B^\top\s+\e_2\in \Z_q^m, \\
\c^*_{2,i} =(\B_i+y_i\cdot\GG)^\top\s+\S_i^\top\e_2 \in \Z_q^{m}, \hspace*{6.6pt}\forall \hspace*{1pt} i\in[\ell].
\end{cases}
\]
\end{enumerate}

\smallskip

\item[Guess:] After being allowed to make additional queries, $\Ad$ outputs $d'\in\{0,1\}$, which is the guess that  the challenge ciphertext $\ct^*$ is an encryption of $M_{d'}$ under $\y_{d'}$ and $\t^*$. Then $\Sd$ computes $c'=d\oplus d'$ and returns it to $\Bd$ as the guess for the bit $c$ chosen by the latter.
\end{description}
\smallskip
Recall that we assume that $\Ad$ breaks the {\sf SR-sA-CPA} security of our SR-PE scheme with probability $\epsilon$, which means
$$\textsf{Adv}_{\Ad}^{\text{\sf SR-sA-CPA}}(\lambda)=\left|\Pr [d'=d\oplus c]-\frac{1}{2}\right|=\epsilon.$$
On the other hand, by construction, we have  $d' = d \oplus c \Leftrightarrow d' \oplus d = c\Leftrightarrow  c' = c$. It then follows that
$$
\textsf{Adv}_{\Sd, \mathsf{PE}}^{\text{\sf wAH-sA-CPA}}(\lambda)  =\left|\Pr [c=c']-\frac{1}{2}\right|=\epsilon/Q.
$$
\qed
\end{proof}

\begin{lemma}\label{lemma: security of type 2}
If there is a PPT Type II adversary $\Ad$ breaking the {\sf SR-sA-CPA} security of our SR-PE scheme with advantage $\epsilon$, then there is a PPT adversary $\Sd$ breaking the {\sf wAH-sA-CPA} security of the AFV PE scheme with the same advantage.
\end{lemma}

\begin{proof}
 Recall that there is an identity $\id^*$ whose predicate is satisfied by $\y_0$ or $\y_1$ and it is not included in $\RL^*$.

 Let $\Bd$ be the challenger in the {\sf wAH-sA-CPA} game for the PE scheme. Algorithm $\Sd$ interacts with $\Ad$ and $\Bd$ as follows.

\begin{description}
\item[Initial:] $\Sd$ first runs $\mathsf{Sys}\left(1^{\lambda}\right)$ to output $\params$. Then $\Ad$ announces to $\Sd$ the target attribute vectors $\y_{0},\y_{1}$ and time $\t^*$. Algorithm  $\Sd$ forwards $\y_{0},\y_{1}$ to $\Bd$.\smallskip
\item[Setup:] $\Sd$ sets an empty revocation list $\RL$ and a binary tree $\BT$ as the sate $\st$. Then $\Sd$ prepares the public parameters as follows:

\smallskip
\begin{enumerate}[leftmargin=*]
\item Receive $\pp_{\PE}=(\B,\hspace{2pt}\{\B_i\}_{i\in [\ell]},\hspace{2pt}\V)$ from~$\Bd$, where $\B,\B_i$ $\in\Z_q^{n\times m},\V\in\Z_q^{n\times \kappa}$.\smallskip
 \item Generate $(\A,\T_{\A})$ by running $\TrapGen(n,q,m)$. Select $\C,\A_i\stackrel{\$}{\leftarrow} \Z_q^{n\times m}$ for each $i\in [\ell]$.
 \item Select $\bar{\S}\stackrel{\$}{\leftarrow} \{-1,1\}^{m\times m}$ and set $\D=\B\bar{\S}-\H(\id^*)\GG$.\smallskip
\item Let the public parameters be $\pp=\left(\A,\hspace{2pt}\B,\hspace{2pt}\C,\hspace{2pt}\D,\hspace{2pt}\{\A_i\}_{i\in [\ell]},\hspace{2pt}\{\B_i\}_{i\in [\ell]},\hspace{2pt}\V\right)$ and send $\pp$ to the adversary $\Ad$.\smallskip
\end{enumerate}

\smallskip
\item[Private Key Oracle:] $\Ad$ is not allowed to issue a private key query for $\id^*$. When $\Ad$ makes a query to $\UserKG(\cdot,\cdot)$ oracle on $(\id,\x)$ such that $\id\neq\id^*$, $\Sd$ returns $\ZZ$ by running $$\mathsf{SampleRight}([\B\mid\B_{\x}],\hspace{2pt} \bar{\S},\hspace{2pt}(\H(\id)-\H(\id^*))\GG,\hspace{2pt}\V,\hspace{2pt}s).$$

\item[Token and Update Key Oracles:] As $\Sd$ knows the master secret key $\T_{\A}$, it can answer all token and update key queries.

\smallskip
\item[Challenge:] $\Ad$ gives two messages $M_0,M_1\in\{0,1\}$ to $\Sd$, who prepares the challenge ciphertext as follows:

\smallskip

\begin{enumerate}[leftmargin=*]
\item Sample $\s\stackrel{\$}{\leftarrow}\Z_q^{n}$, $\e_1{\hookleftarrow}\chi^{m}$. Choose $\bar{\R},\R_i\xleftarrow{\$}\{-1,1\}^{m\times m}$ for each $i\in [\ell]$. \smallskip
\item Pick $d\stackrel{\$}{\leftarrow}\{0,1\}$ and set $M'_0=M_d, M'_1=M_{1\oplus d}$. Forward $M'_0,M'_1$ as two challenge messages to the PE challenger $\Bd$. The latter chooses $c\stackrel{\$}{\leftarrow}\{0,1\}$ and returns $(\c',\hspace{2pt}\c'_0,\hspace{2pt}\{\c'_{i}\}_{i\in [\ell]})$ as a PE encryption of $M'_c$ under $\y_c$. \smallskip
\item Output $\ct^*=(\c^*,\hspace{2pt}\c^*_{1},\hspace{2pt}\{\c^*_{1,i}\}_{i\in[\ell]},\hspace{2pt}\c^*_{1,0},\hspace{2pt}\c^*_{2},\hspace{2pt}\{\c^*_{2,i}\}_{i\in[\ell]})$ as an SR-PE encryption of $M_d$ under $\y_d,\t^*$, where:
\[
\begin{cases}
\c^* =\c' \in\Z_q^{\kappa}, \\[1pt]
\c^*_1 = \A^\top\s+\e_1 \in \Z_q^m, \\[1pt]
\c^*_{1,i} = (\A_i+y_i\cdot\GG)^\top\s+\R_i^\top\e_1 \in \Z_q^{m}, \hspace*{6.6pt}\forall \hspace*{1pt} i\in[\ell]\\[1pt]
\c^*_{1,0} = \C_\t^\top\s+\bar{\R}^\top\e_1 \in \Z_q^{m},\\[1pt]
\c^*_2 = \c'_0\in \Z_q^m, \\[1pt]
\c^*_{2,i} =\c'_{i} \in \Z_q^{m}, \hspace*{6.6pt}\forall \hspace*{1pt} i\in[\ell].
\end{cases}
\]
\end{enumerate}

\smallskip

\item[Guess:] After being allowed to make additional queries,~$\Ad$ outputs $d'\in\{0,1\}$, which is the guess that  the challenge ciphertext $\ct^*$ is an encryption of $M_{d'}$ under $\y_{d'}$ and $\t^*$. Then $\Sd$ computes $c'=d\oplus d'$ and returns it to $\Bd$ as the guess for the bit $c$ chosen by the latter.
\end{description}
\smallskip
Recall that we assume that $\Ad$ breaks the {\sf SR-sA-CPA} security of our SR-PE scheme with probability $\epsilon$, which means
$$
\textsf{Adv}_{\Ad}^{\text{\sf SR-sA-CPA}}(\lambda)=\left|\Pr [d'=d\oplus c]-\frac{1}{2}\right|=\epsilon.
$$
By construction, we have  $d' = d \oplus c \Leftrightarrow d' \oplus d = c\Leftrightarrow  c' = c$. It then follows that
$$
\textsf{Adv}_{\Sd, \mathsf{PE}}^{\text{\sf wAH-sA-CPA}}(\lambda)  =\left|\Pr [c=c']-\frac{1}{2}\right|=\epsilon.
$$
\qed
\end{proof}

\smallskip

Finally, recall that algorithm $\Sd$ can guess the type of the adversary correctly with probability $1/2$ and the adversary's behaviour is independent from the guess. It then follows from the results of Lemma~\ref{lemma: security of type 1} and  Lemma~\ref{lemma: security of type 2} that
\[
\textsf{Adv}_{\Ad}^{\text{\sf SR-sA-CPA}}(\lambda) = \frac{1}{2}\Big(\frac{1}{Q}\textsf{Adv}_{\Sd, \mathsf{PE}}^{\text{\sf wAH-sA-CPA}}(\lambda) + \textsf{Adv}_{\Sd, \mathsf{PE}}^{\text{\sf wAH-sA-CPA}}(\lambda) \Big).
\]
By Theorem~\ref{theorem:PE}, we then have that $\textsf{Adv}_{\Ad}^{\text{\sf SR-sA-CPA}}(\lambda) = \mathrm{negl}(\lambda)$, provided that the $(n,q,\chi)$-LWE assumption holds. This concludes the proof.
\qed
\end{proof}

\section{Conclusion and Open Problems} \label{section: collusion}
We introduced the server-aided revocation mechanism in the setting of predicate encryption and then gave a lattice-based instantiation. We proved that the scheme is selectively secure based on the \textsf{LWE} assumption.  Achieving the stronger adaptive security notion seems to require that the underlying PE be adaptively secure. However, to the best of our knowledge, existing lattice-based PE schemes~\cite{AgrawalFV11,Xagawa13,GayMW15,GorbunovVW15} only achieved selective security. We therefore view the problem of constructing adaptively secure lattice-based SR-PE as an interesting open question.
Another question that we left unsolved is to investigate whether our
design approach (i.e., combining two PE instances, one IBE instance and the CS method) would yield a generic construction for SR-PE.
%Another limitation of our scheme is that the size of ciphertext, which could be $O(1)$ in the best case, could also be $O(N/2)$ in the worst case. We left open the problem of constructing a secure lattice-based RPE with better ciphertext size (e.g., linear in the number of revoked users).

\smallskip
\noindent
{\sc Acknowledgements. }
We thank the reviewers for helpful discussions and comments.
 The research was supported by the ``Singapore Ministry of Education under Research Grant  MOE2016-T2-2-014(S)''.

\end{document}